\documentclass[12pt]{article}
\usepackage{epsfig,amsmath,amssymb,amsfonts,amstext,amsthm,mathrsfs}
\usepackage{latexsym,graphics,epsf,epsfig,psfrag}
\usepackage{cite}
\usepackage{color}

\usepackage{caption}
\usepackage{subcaption}
\usepackage{graphicx}

\makeatletter

\renewcommand\p@subfigure{\thefigure\,}
\DeclareCaptionLabelFormat{mysublabelfmt}{(\alph{sub\@captype})\,}
\makeatother

\newtheorem{corollary}{\textbf{Corollary}}
\newtheorem{lemma}{\textbf{Lemma}}
\newtheorem{theorem}{\textbf{Theorem}}
\newtheorem{proposition}{\textbf{Proposition}}
\newtheorem{remark}{\textbf{Remark}}

\newcommand{\nn}{\nonumber}

\newcommand{\cX}{\mathcal{X}}

\newcommand{\cY}{\mathcal{Y}}

\newcommand{\cW}{\mathcal{W}}
\newcommand{\cS}{\mathcal{S}}

\newcommand{\cN}{\mathcal{N}}

\DeclareMathAlphabet{\matheuf}{U}{euf}{m}{n}

\begin{document}
\vspace*{-2cm}

\begin{center}
  \baselineskip 1.3ex {\Large \bf Parallel Gaussian Networks with a Common State-Cognitive Helper
  \footnote{The material in this paper was presented
in part at the IEEE Information Theory Workshop, Seville, Spain, September 2013.} \footnote{The work of R. Duan and Y. Liang was supported by the National Science Foundation under Grants CCF-10-26566 and CCF-12-18451 and by the National Science Foundation CAREER Award under Grant
CCF-10-26565. The work of A. Khisti was supported by the Canada Research Chair's Program. The work of S. Shamai(Shitz) was supported by the Israel Science Foundation (ISF), and the European Commission in the framework of the Network of Excellence in FP7 Wireless COMmunications NEWCOM$\#$.}\\
}
 \vspace{0.15in} Ruchen Duan, Yingbin Liang,
\footnote{Ruchen Duan and Yingbin Liang are with the Department of Electrical
Engineering and Computer Science, Syracuse University, Syracuse, NY 13244 USA (email: \{yliang06,rduan\}@syr.edu).}
Ashish Khisti,\footnote{Ashish Khisti is with the Department of Electrical and Computer Engineering, University of Toronto, Toronto, ON, M5S3G4, Canada (email: akhisti@comm.utoronto.ca).}
Shlomo Shamai (Shitz)\footnote{Shlomo Shamai (Shitz) is with the Department of Electrical Engineering, Technion-Israel Institute of Technology, Technion city, Haifa 32000, Israel (email: sshlomo@ee.technion.ac.il).}
\end{center}

\begin{abstract}
A class of state-dependent parallel networks with a common state-cognitive helper, in which $K$ transmitters wish to send $K$ messages to their corresponding receivers over $K$ state-corrupted parallel channels, and a helper who knows the state information noncausally wishes to assist these receivers to cancel state interference. Furthermore, the helper also has its own message to be sent simultaneously to its corresponding receiver. Since the state information is known only to the helper, but not to the corresponding transmitters $1,\dots,K$, transmitter-side state cognition and receiver-side state interference are mismatched. Our focus is on the high state power regime, i.e., the state power goes to infinity.
Three (sub)models are studied. Model I serves as a basic model, which consists of only one transmitter-receiver (with state corruption) pair in addition to a helper that assists the receiver to cancel state in addition to transmitting its own message. Model II consists of two transmitter-receiver pairs in addition to a helper, and only one receiver is interfered by a state sequence. Model III generalizes model I include multiple transmitter-receiver pairs with each receiver corrupted by independent state. For all models, inner and outer bounds on the capacity region are derived, and comparison of the two bounds leads to characterization of either full or partial boundary of the capacity region under various channel parameters.
\end{abstract}

\section{Introduction}\label{sec:Introduction}

State-dependent network models have recently caught intensive attention. In these models, receivers are interfered by random state sequences, and some or all of the transmitters know the corresponding state sequences that interfere their targeted receivers noncausally, and exploit dirty paper coding to assist the receivers to cancel the state interference. For example, the state-dependent broadcast channel has been studied in, e.g., \cite{Steinberg05,Lapidoth11}, the state-dependent multiple access channel (MAC) has been studied in, e.g., \cite{Somekh08MAC,Laneman08,Wang11hsiang}, the state-dependent relay channel has been studied in, e.g., \cite{Aref09,Zaidi11}, and the state-dependent interference channel has been studied in, e.g., \cite{Zhang11a,Duan12TIT,Duan13ISIT,Ghas13}.



In this paper, we study a class of state-dependent parallel networks with a common state-cognitive helper (see Figure \ref{fig:channelmodel}), in which $K$ transmitters wish to send $K$ messages to their corresponding receivers over $K$ state-corrupted parallel channels, and a helper who knows the state information noncausally wishes to assist these receivers to cancel state interference. Furthermore, the helper also has its own message to be sent simultaneously to its corresponding receiver. Since the state information is known only to the helper, but not to the corresponding transmitters $1,\dots,K$, transmitter-side state cognition and receiver-side state interference are mismatched. Our goal is to investigate such a mismatched scenario in high state power regime, i.e., as the power of the state sequences go to infinity. This model is well justified in practical wireless networks. For example, in a cellular network, a base station likely causes interference to receivers in its adjacent cells, and such interference can be treated as state known at this base station. The base station can then serve as a helper to assist the receivers to cancel state interference, which is particularly desirable when the state power is large. Such a model suggests to exploit the state cognition for improving communication rates other than the traditional message cognition studied in the context of cognitive channels and networks.



This network model has a few properties that differentiate it from previous studies of state-dependent networks. In this model, the state knowledge is known only to the helper, which does not know the messages that it assists for transmission. This is different from the classic state-dependent channel in \cite{GPEncoding} and most of its followups, in which the transmitter knows both the message and the state. Although such a mismatch structure appeared also in some previously studied models such as the state-dependent multiple-access channel in \cite{Laneman08}, and the relay channel in \cite{Aref09,Zaidi11}, the structure of multiple state-interfered receivers differentiates our model from these studies. Since our model has the nature of compound state interference at $K$ receivers (i.e., the $K$ receivers are corrupted by different states), the helper's assistance scheme needs to trade off among the $K$ receivers' performances.

In this paper, we study three (sub)models of the state-dependent parallel networks with a common helper. Model I serves as a basic model, which consists of only one state-corrupted receiver ($K=1$) and a helper that assists this receiver to cancel state interference in addition to transmitting its own message. Our study of this model provides necessary techniques to deal with state in the mismatched context for studying more complicated models II and III. In fact, this model can be viewed as the state-dependent Z-interference channel, in which the interference is only at receiver 1 caused by the helper. In contrast to the state-dependent Z-interference channel studied previously in \cite{Hajizadeh13}, which assumes that state interference at both receivers are known to both (corresponding) transmitters, our model assumes that state interference is known noncausally only to the helper, not to the corresponding transmitter 1.

In general, it is challenging to design capacity-achieving schemes for such a system with mismatched property. Clearly, it is not possible for transmitter 1 to directly cancel state interference due to the large state power. One natural idea is to apply lattice coding in high state power regime as in \cite{Ashish07} for the state-dependent multiple access channel (MAC). However, lattice coding does not achieve the capacity for our model here. Another approach is to apply dirty paper coding \cite{Costa83}. However, the difficulty here lies in that the helper needs to resolve the tension between transmitting its own message and helping receiver 1 to cancel its interference.

In this paper, we design a layered coding scheme, in which a dirty paper coding scheme for the helper to assist state cancelation is superposed with the helper's transmission to its own message. Due to mismatched state cognition and interference, in our dirty paper coding scheme, correlation between the state variable and the state-cancelation variable is a design parameter, and can be chosen to optimize the rate region. This is in contrast to classical dirty paper coding \cite{Costa83}, in which such a correlation parameter is fixed for fully canceling the state. Based on such a layered coding scheme, we derive achievable regions for both the discrete memoryless and Gaussian channels. We further derive an outer bound for the Gaussian channel in high state power regime. By comparing the inner and outer bounds, we characterize the boundary of the capacity region either fully or partially for all Gaussian channel parameters in high state power regime. Our result also implies that the capacity region is strictly inside the capacity region of the corresponding channel without state \cite{Sason04}. This is in contrast to the results for Costa type of dirty paper channels, for which dirty paper coding achieves the capacity of the corresponding channels without state.







We then further study model II, which consists of two transmitter-receiver pairs in addition to the helper, and only one receiver is interfered by a state sequence. Here, the challenge lies in the fact that the helper inevitably causes interference to receiver 2 while assisting receiver 1 to cancel the state. For this model, we start with the scenario with the helper fully assisting the receivers without transmitting its own message. We first derive an outer bound on the capacity region. We then develop a two-layer dirty paper coding scheme with one layer helping receiver 1 to cancel state via dirty paper coding, and with the other layer of dirty paper coding canceling the interference caused by the helper in assisting receiver 1. By comparing inner and outer bounds, we characterize two segments of the capacity region boundary. One segment corresponds to the case, in which our scheme achieves the point-to-point channel capacity for receiver 2 and certain positive rate for receiver 1. This implies that the helper is able to assist receiver 1 without causing interference to receiver 2 effectively. The other segment corresponds to the case, in which our scheme achieves the best single-user rate for receiver 1 with assistance of the helper, while receiver 2 treats the helper's signal as noise. Such a scheme is guaranteed by our outer bound to be the best to achieve the sum capacity under certain channel parameters. We further extend these results to the scenario with the helper sending its own message in addition to assisting the two receivers.

We finally study model III, in which a common helper assists multiple transmitter-receiver pairs with each receiver corrupted by an independently distributed state sequence. We note that this model is more general than model I, but does not include model II as a special case. This is because model III has each receiver (excluding the helper) being corrupted by an infinitely powered state sequence, and hence never reduces to the model II, in which receiver 2 is not corrupted by a state sequence. This also leads to different technical challenges to characterize the capacity for model III due to the compound state interference. The same technical challenge is also reflected in the studies \cite{Mitran06,Ashish07,Piantanida10} of the state-dependent compound channel, for which the capacity is not known in general. As for model II, we also start with the scenario, in which the helper fully assists other users without sending its own message. We first derive a useful outer bound, which captures the sum rate limit due to the common helper. We then derive an inner bound based on a time-sharing scheme, in which the helper alternatively assists receivers. Somewhat interestingly, such a time-sharing scheme achieves the sum capacity under many channel parameters, although each individual transmitter may not be able to achieve its individual best rate. This is because these transmitters effectively have larger power during their transmissions in the time-sharing scheme so that the transmission rate matches the outer bound on the sum rate. We also characterize the full capacity region under certain channel parameters. We then extend our results to the general scenario with the helper also transmitting its own message.

The rest of the paper is organized as follows. In Section \ref{sec:ChannelModel}, we describe the channel model. In Sections \ref{sec:ResultSingle}, \ref{sec:Same}, and \ref{sec:Independent}, we present our results for models I, II, and III, respectively. Finally, in Section \ref{sec:Conclusion}, we conclude the paper with a few remarks.

\section{Channel Model}\label{sec:ChannelModel}

\begin{figure}[thb]
\centering
\includegraphics[width=3.in]{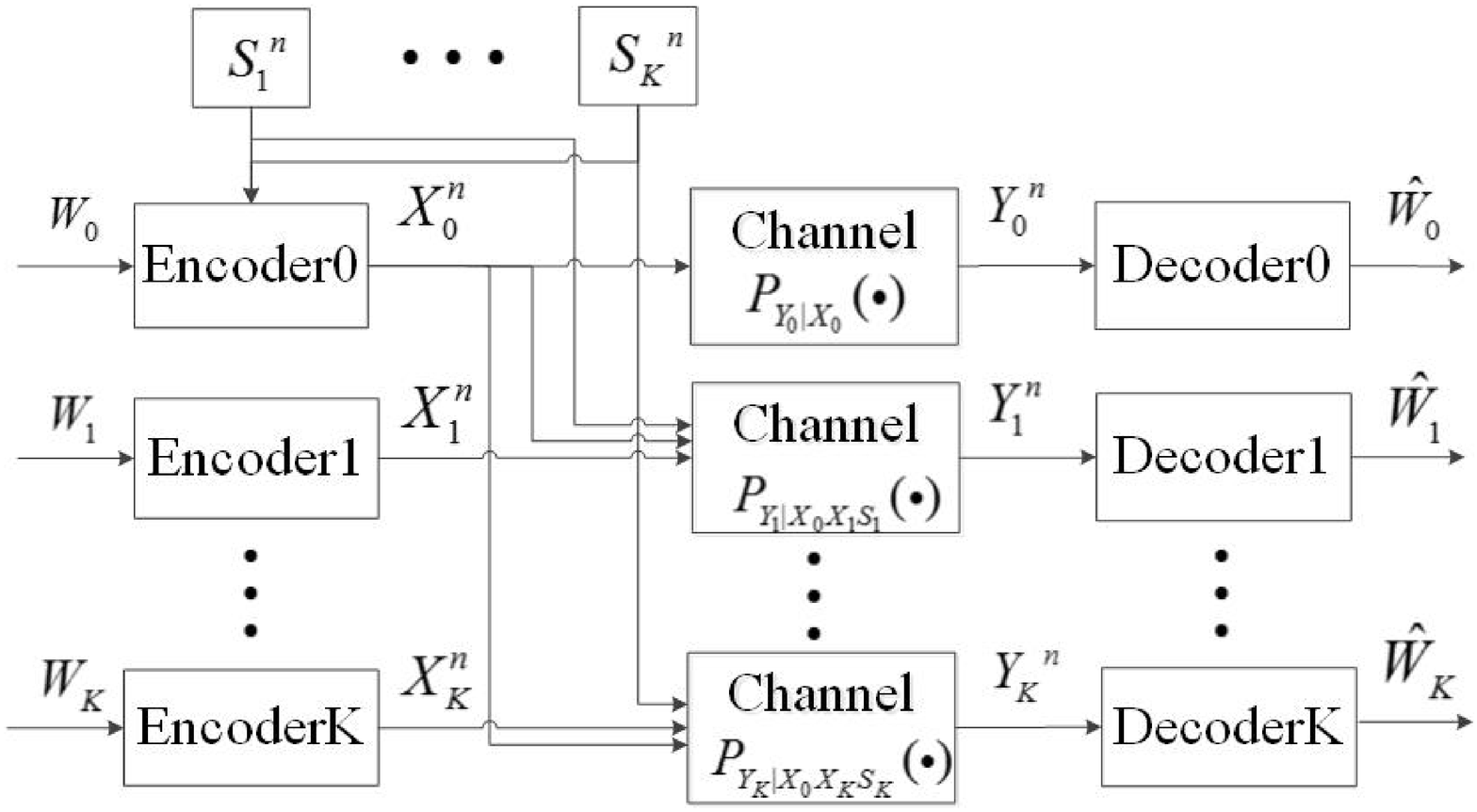}
\caption{Parallel Gaussian channel model with a common state-cognitive helper.}\label{fig:channelmodel}
\end{figure}

In this paper, we investigate the state-dependent parallel network with a common state-cognitive helper (see Figure \ref{fig:channelmodel}), in which $K$ transmitters wish to send $K$ messages to their corresponding receivers over state-corrupted parallel channels, and a helper who knows the state information noncausally wishes to assist these receivers to cancel state interference. Furthermore, the helper also has its own message to be sent simultaneously to its corresponding receiver.

More specifically, each transmitter (say transmitter $k$) has an encoder $f_k:\cW_k \rightarrow \cX_k^n$, which maps a message $w_k\in \cW_k$ to a codeword $x_k^n\in \cX_k^n$ for $k=1, \ldots, K$. The $K$ inputs $x_1^n, \ldots, x_K^n$ are transmitted over $K$ parallel channels, respectively. Each receiver (say receiver $k$) is interfered by an i.i.d.\ state sequence $S_K^n$ for $k=1,\ldots, K$, which is unknown at none of transmitters $1,\ldots, K$ and receivers $1,\ldots, K$. A common helper (referred to as transmitter 0) is assumed to know all state sequences $S_K^n$ for $k=1,\ldots, K$ noncausally. Thus, the encoder at the helper, $f_0:\cW_0 \times \{\cS_1^n,\dots,\cS_K^n\} \rightarrow \cX_0^n$, maps a message $w_0\in \cW_0$ and the state sequences $(s_1^n,\ldots,s_K^n) \in \cS_1^n\times \ldots\times\cS_K^n$ to a codeword $x_0^n\in \cX_0^n$. The entire channel transition probability is given by $P_{Y_0\vert X_{0}}\prod_{k=1}^K P_{Y_k\vert X_{0},X_{k},S_k}$. There are $K=1$ decoders with each at one receiver, $g_k:\cY_k^n\rightarrow \cW_k$, maps a received sequence $y_k^n$ into a message $\hat{w}_k\in\cW_k$ for $k=0,1,\ldots,K$.
\begin{remark}
Without state interference, our model becomes the $K+1$-user Z-interference channel, in which the signal of transmitter 0 interferences all remaining $K$ receivers.
\end{remark}

The average probability of error for a length-$n$ code is defined as
\begin{flalign}\label{PE}
P_e^{(n)} = & \frac{1}{|\cW_0||\cW_1|\dots|\cW_K|}\sum_{w_0=1}^{|\cW_0|} \sum_{w_1=1}^{|\cW_1|}\dots \sum_{w_k=1}^{|\cW_k|} Pr\lbrace(\hat{w}_0,\hat{w}_1,\dots,\hat{w}_K) \neq (w_0,w_1,\dots,w_K)\rbrace.
\end{flalign}
A rate tuple $(R_0,R_1,\dots,R_K)$ is {\em achievable} if there exists a sequence of message sets $\cW_{k}^{(n)}$ with $|\cW_{k}^{(n)}|=2^{nR_k}$ for $k=0,\,1,\dots,K$, and encoder-decoder tuples $(f_{0}^{(n)},f_{1}^{(n)},\dots,f_{K}^{(n)}, g_{0n},g_{1n},\dots,g_{Kn})$ such that the average error probability $P_e^{(n)} \rightarrow 0$ as $n \to \infty$. The {\em capacity region} is defined to be the closure of the set consists of all achievable rate pairs $(R_0,R_1,\dots,R_K)$.


In this paper, we study the following three Gaussian channel models.

In model I, $K=1$. The channel outputs at receiver 0 and 1 for one symbol time are given by
    \begin{subequations}
  \begin{flalign}
&Y_0=X_0+N_0,\label{eq:ChannelModelI-0}\\
&Y_1=X_0+X_1+S_1+N_1. \label{eq:ChannelModelI-1}
\end{flalign}
  \end{subequations}

  In model II, $K=2$, and the channel outputs at receivers 0, 1 and 2 for one symbol time are given by
    \begin{subequations}
  \begin{flalign}
&Y_0=X_0+N_0,\label{eq:ChannelModelII-0}\\
&Y_1=X_0+X_1+S_1+N_1, \label{eq:ChannelModelII-1}\\
&Y_2=X_0+X_2+N_2. \label{eq:ChannelModelII-2}\\
\end{flalign}
  \end{subequations}

  In model III, $K$ is general and the channel outputs at receivers 0 and receivers $1,\dots,K$ for one symbol time are given by
  \begin{subequations}
  \begin{flalign}
&Y_0=X_0+N_0,\label{eq:ChannelModelIII-0}\\
&Y_k=X_0+X_k+S_k+N_k, \quad\quad k=1,\dots,K \label{eq:ChannelModelIII-1}
\end{flalign}
  \end{subequations}
  In the above three models, the noise variables $N_0,\,N_1\dots, N_K$ and the state variable $S_1,\dots,S_K$ are Gaussian distributed with distributions $N_0, \dots,N_K \sim \mathcal{N}(0,1)$ and $S_k \sim \mathcal{N}(0,Q_k)$ for $k=1,\dots,K$, and all of the variables are independent and are i.i.d.\ over channel uses. The channel inputs $X_0,\,X_1,\dots,X_K$ are subject to the average power constraints $\frac{1}{n}\sum_{i=1}^n X_{ki}^2 \leqslant P_k$ for $k=0,\,1,\dots,K$.

We are interested in the regime of high state power, i.e., as $Q_k\rightarrow \infty$ for $k=1,\dots,K$. Our goal is to characterize the capacity region of the Gaussian channels in this regime.

\section{Model I: $K=1$}\label{sec:ResultSingle}

Model I with $K=1$ is a basic model, in which the helper assists one transmitter-receiver pair. Understanding this model will help the study of the general parallel network. In this section, we first develop inner and outer bounds on the capacity region, and then characterize the boundary of the capacity region based on these bounds.


\subsection{Inner Bound}\label{sec:InnerSingle}

The major challenge in designing an achievable scheme arises from the mismatched property due to transmitter-side state cognition and receiver-side state interference, i.e., state interference to receiver 1 is known noncausally only to transmitter 0 (the helper), not to the corresponding transmitter 1. Since we study the regime with large state power, transmitter 1 can send information to receiver 1 only if the helper assists to cancel the state. Thus, the helper needs to resolve the tension between transmitting its own message to receiver 0 and helping receiver 1 to cancel its interference. A simple scheme of time-sharing between the two transmitters in general is not optimal.

We design a layered coding scheme as follows. The helper splits its signal into two parts in a layered fashion: one (represented by $X_0'$ in Lemma \ref{th:InnerDMCSingle}) for transmitting its own message and the other (represented by $U$ in Lemma \ref{th:InnerDMCSingle}) for helping receiver 1 to remove both state and signal interference. In particular, the second part of the scheme applies a single-bin dirty paper coding scheme, in which transmission of $W_1$ and treatment of state interference for decoding $W_1$ are performed separately by transmitters 1 and 0. This is because the helper knows the state but does not know the message (of transmitter 1) that the state interferes, and hence cannot encode this message via the regular multi-bin dirty paper coding as in \cite{Costa83}. Based on such a scheme, we obtain the following achievable rate region for the discrete memoryless channel, which is useful for deriving an inner bound for the Gaussian channel.
\begin{lemma}\label{th:InnerDMCSingle}
For the discrete memoryless model I, an inner bound on the capacity region consists of rate pairs $(R_0, R_1)$ satisfying:
\begin{subequations}
\begin{flalign}
  &R_0 \leqslant I(X_0';Y_0)  \\
&R_1 \leqslant I(X_1;Y_1|U)  \\
&R_1 \leqslant I(X_1U;Y_1) - I(U;S_1X_0')\label{eq:InnerDMCSinglere}
\end{flalign}
\end{subequations}
for some distribution $P_{S_1}P_{X_0'}P_{U|S_1X_0'}P_{X_0|US_1X_0'}P_{X_1}P_{Y_0Y_1|S_1X_0X_1}$.
\end{lemma}
\begin{proof}
  The proof is detailed in Appendix \ref{apx:InnerDMCSingle}.
\end{proof}

Based on Lemma \ref{th:InnerDMCSingle}, we have the following simpler inner bound by adding a constraint to remove \eqref{eq:InnerDMCSinglere} as a redundant bound.
\begin{corollary}\label{th:InnerDMCCor}
For the discrete memoryless model I, an inner bound on the capacity region consists of rate pairs $(R_0, R_1)$ satisfying:
\begin{subequations}
  \begin{flalign}
&R_0 \leqslant I(X_0';Y_0)\\
& R_1 \leqslant I(X_1;Y_1|U) \label{eq:model1r1}
\end{flalign}
\end{subequations}
for some distribution $P_{S_1}P_{X_0'}P_{U|S_1X_0'}P_{X_0|US_1X_0'}P_{X_1}P_{Y_0Y_1|S_1X_0X_1}$ that satisfies
\begin{equation}
  I(U;Y_1)\geqslant I(U;S_1X_0').
\end{equation}
\end{corollary}
The inner bound in Corollary \ref{th:InnerDMCCor} corresponds to an intuitive achievable scheme based on successive cancelation. Namely, the condition guarantees that receiver 1 decodes the auxiliary random variable $U$ first, and then removes it from its output and decodes the message, which results in the bound \eqref{eq:model1r1}. In particular, cancelation of $U$ leads to cancelation of state interference at receiver 1.

We next derive an inner bound for the Gaussian channel of model I based on Corollary \ref{th:InnerDMCCor}.
\begin{proposition}\label{th:InnerGaussianSingle}
For the Gaussian channel of model I, an inner bound on the capacity region consists of rate pairs $(R_0, R_1)$ satisfying:
\begin{subequations}
  \begin{flalign}
  &R_0 \leqslant  \frac{1}{2}\log\left( 1+\frac{\bar{\beta} P_0}{\beta P_0 + 1}\right) \label{eq:InnerGaussian-1}\\
  &R_1 \leqslant  \frac{1}{2}\log\left( 1+\frac{P_1}{1+(1-\frac{1}{\alpha})^2 \beta P_0}\right)\label{eq:InnerGaussian-2}
\end{flalign}
\end{subequations}
for some real constants $\alpha \geqslant 0$ and $0 \leqslant \beta \leqslant 1$ that satisfy
\begin{equation}
  \alpha^2(\bar{\beta}P_0+Q_1)(\beta P_0 +P_1 +1)- 2\alpha\beta P_0 (\bar{\beta}P_0+Q_1) -\beta^2 P_0^2 \leqslant 0.
\end{equation}
As $Q_1\rightarrow \infty$, the preceding condition becomes $\alpha \leqslant  \frac{2\beta P_0}{\beta P_0+P_1 +1}$.
\end{proposition}

\begin{proof}
Proposition \ref{th:InnerGaussianSingle} follows from Corollary \ref{th:InnerDMCCor} by choosing the joint Gaussian distribution for random variables as follows:
  \begin{flalign}
  U&= X_0''+\alpha (S_1+X_0'),\;X_0=X_0'+X_0''\nn\\
  X_0'& \sim \cN(0,\beta P_0),\; X_0'' \sim \cN(0,\bar{\beta} P_0)\nn\\
  X_{1}&\sim \mathcal{N}(0,P_{1})\nn
\end{flalign}
where $X_0'$, $X_0''$, $X_1$ and $S_1$ are independent, $\alpha\geqslant 0$, $0 \leqslant \beta \leqslant 1$, and $\bar{\beta}= 1-\beta$.
\end{proof}

We note that in Proposition \ref{th:InnerGaussianSingle}, the parameter $\alpha$ captures correlation between the state variable $S$ and the auxiliary variable $U$ for dealing with the state, and can be chosen to optimize the rate region. This is in contrast to the classical dirty paper coding \cite{Costa83}, in which such correlation parameter is fixed for state cancelation. Therefore, although Corollary \ref{th:InnerDMCCor} may provide a smaller inner bound than that given in Lemma \ref{th:InnerDMCSingle}, it can be shown that two inner bounds are equivalent for our chosen auxiliary random variables and input distribution after optimizing over $\alpha$.

\subsection{Outer Bound}\label{sec:OuterSingle}

In this subsection, we provide an outer bound on the capacity region in high state power regime, i.e., as $Q\rightarrow \infty$.
\begin{proposition}\label{th:OuterGaussianSingle}
For the Gaussian channel of model I, an outer bound on the capacity region for the regime when $Q_1\rightarrow \infty$ consists of rate pairs $(R_0, R_1)$ satisfying:
\begin{subequations}
  \begin{flalign}
      R_1 \leqslant \frac{1}{2} \log(1+P_1) \label{eq:outer-1} \\
      R_0+R_1 \leqslant \frac{1}{2} \log(1+P_0). \label{eq:outer-2}
    \end{flalign}
\end{subequations}
\end{proposition}
The bound \eqref{eq:outer-1} on $R_1$ follows simply from the capacity of the point-to-point channel between transmitter 1 and receiver 1 without signal and state interference. The bound \eqref{eq:outer-2} on the sum rate is limited only by the power $P_0$ of the helper, and does not depend on the power $P_1$ of transmitter 1. Intuitively, this is because $P_0$ is split for transmission of $W_0$ and for helping transmission of $W_1$ by removing state interference, and hence $P_0$ determines a trade-off between $R_0$ and $R_1$. On the other hand, improving the power $P_1$, although may improve $R_1$, can also cause more interference for receiver 1 to decode the auxiliary variable for canceling state and interference. Thus, the balance of the two effects determine that $P_1$ does not affect the sum rate.
\begin{proof}
  The proof is detailed in Appendix \ref{apx:OuterGaussianSingle}.
\end{proof}

We further note that although the sum-rate upper bound \eqref{eq:outer-2} can be achieved easily by keeping transmitter 1 silent (i.e., $R_0$ achieves the sum rate bound with $R_1=0$), we are interested in characterizing the capacity region (i.e., the trade-off between $R_0$ and $R_1$) rather than a single point that achieves the sum-rate capacity. In the next section, we characterize such optimal trade-off based on the sum-rate bound.
\begin{remark}
The outer bound in Proposition \ref{th:OuterGaussianSingle} is strictly inside an achievable rate region of the corresponding channel without state interference (i.e., the Z-interference channel) \cite{Sason04}. This implies that the capacity region of our model is strictly inside that of the corresponding channel without state. This suggests that state interference does cause performance degradation for systems with mismatched state cognition and interference in high state power regime. This is in contrast to the results on Costa-type dirty paper channels \cite{Costa83}, for which dirty paper coding achieves the capacity of the corresponding channels without state.
\end{remark}

\subsection{Capacity Region}\label{sec:CapacitySingle}
In this section, we characterize the boundary points of the capacity region for the Gaussian channel of model I based on the inner and outer bounds given in Propositions \ref{th:InnerGaussianSingle} and \ref{th:OuterGaussianSingle}, respectively. We partition the Gaussian channel into three cases based on the conditions on the power constraints: (1) $P_1 \geqslant P_0 + 1$; (2) $P_0 -1\leqslant P_1< P_1 +1$ and (3) $0 \leqslant P_1 <P_0-1$. For each case, we optimize the dirty paper coding parameter $\alpha$ that satisfies $0 \leqslant \alpha \leqslant \frac{2\beta P_0}{\beta P_0 +P_1 +1}$ to find achievable rate points that lie on the sum-rate upper bound \eqref{eq:outer-2} in order to characterize the boundary points of the capacity region. 

{\em Case 1: $P_1 \geqslant P_0 + 1$.} The capacity region is fully characterized as illustrated in Fig.~\ref{fig:SingleCase1}.

\begin{figure}[thb]
\centering
\includegraphics[width=2.8in]{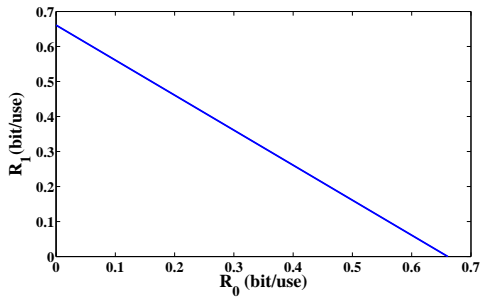}
\caption{The capacity region for case 1 with $P_0=1.5$ and $P_1=3$.}\label{fig:SingleCase1}
\end{figure}

\begin{theorem}\label{th:capacity-1}
For the Gaussian channel of model I in the regime when $Q_1\rightarrow \infty$, if $P_1 \geqslant P_0 + 1$, the capacity region consists of the rate pairs $(R_0,R_1)$ satisfying
\begin{flalign}
    R_0 + R_1&\leqslant \frac{1}{2} \log(1+P_0).
\end{flalign}
\end{theorem}

\begin{proof}
Let $\tilde{P}_1$ be the actual transmission power for transmitting $W_1$. Then the inner bound \eqref{eq:InnerGaussian-2} on $R_1$ is optimized when $\alpha= \frac{2\beta P_0}{\beta P_0 +\tilde{P}_1 +1}$, By setting $\tilde{P}_1=\beta P_0+1$, the inner bound given in Proposition \ref{th:InnerGaussianSingle} matches the outer bound given in Proposition \ref{th:OuterGaussianSingle}, and hence is the capacity region.
\end{proof}

Theorem \ref{th:capacity-1} implies that when $P_1$ is large enough, the power of the helper limits the system performance. Furthermore, since $P_1$ for transmission of $W_1$ causes interference to receiver 1 to decode the auxiliary variable for state and interference cancelation, beyond a certain value, increasing $P_1$ does not improve the rate region any more. Theorem \ref{th:capacity-1} also suggests that in order to achieve different points on the boundary of the capacity region (captured by parameters $\beta$), different amounts of power $\tilde{P_1}$ should be applied.

{\em Case 2: $P_0-1 \leqslant P_1 < P_0 +1$.} For this case, if $P_1 \geqslant 1$, i.e., $P_1$ is larger than the noise power, inner and outer bounds match over the line between the points A and B as illustrated in Fig.~\ref{fig:SingleCase2} (a) and (c), and thus optimal trade-off between $R_0$ and $R_1$ is achieved over the points on this line. If $P_1 < 1$, the inner and outer bounds match only at the rate point A as illustrated in Fig.~\ref{fig:SingleCase2} (b) and (d), which achieves the sum-rate capacity. We note that Fig.~\ref{fig:SingleCase2} (a) and Fig.~\ref{fig:SingleCase2} (b) are different from Fig.~\ref{fig:SingleCase2} (c) and Fig.~\ref{fig:SingleCase2} (d) in the outer bound. Fig.~\ref{fig:SingleCase2} (c) and Fig.~\ref{fig:SingleCase2} (d) correspond to the case with $P_0\ge P_1$, and hence the capacity region is also upper bounded by the point-to-point capacity of $R_1$. Such a bound is redundant in Fig.~\ref{fig:SingleCase2} (a) and Fig.~\ref{fig:SingleCase2} (b) which correspond to the case with $P_0<P_1 $, because $P_0$ is not large enough to perfectly cancel state and signal interference at receiver 1. However, in case 3, we show that this point-to-point capacity of $R_1$ is achievable simultaneously with a certain positive $R_0$. We summarize the above capacity result in the following theorem.
\begin{figure*}[hbt!]
\begin{center}
\begin{tabular}{cc}
\includegraphics[width=2.8in]{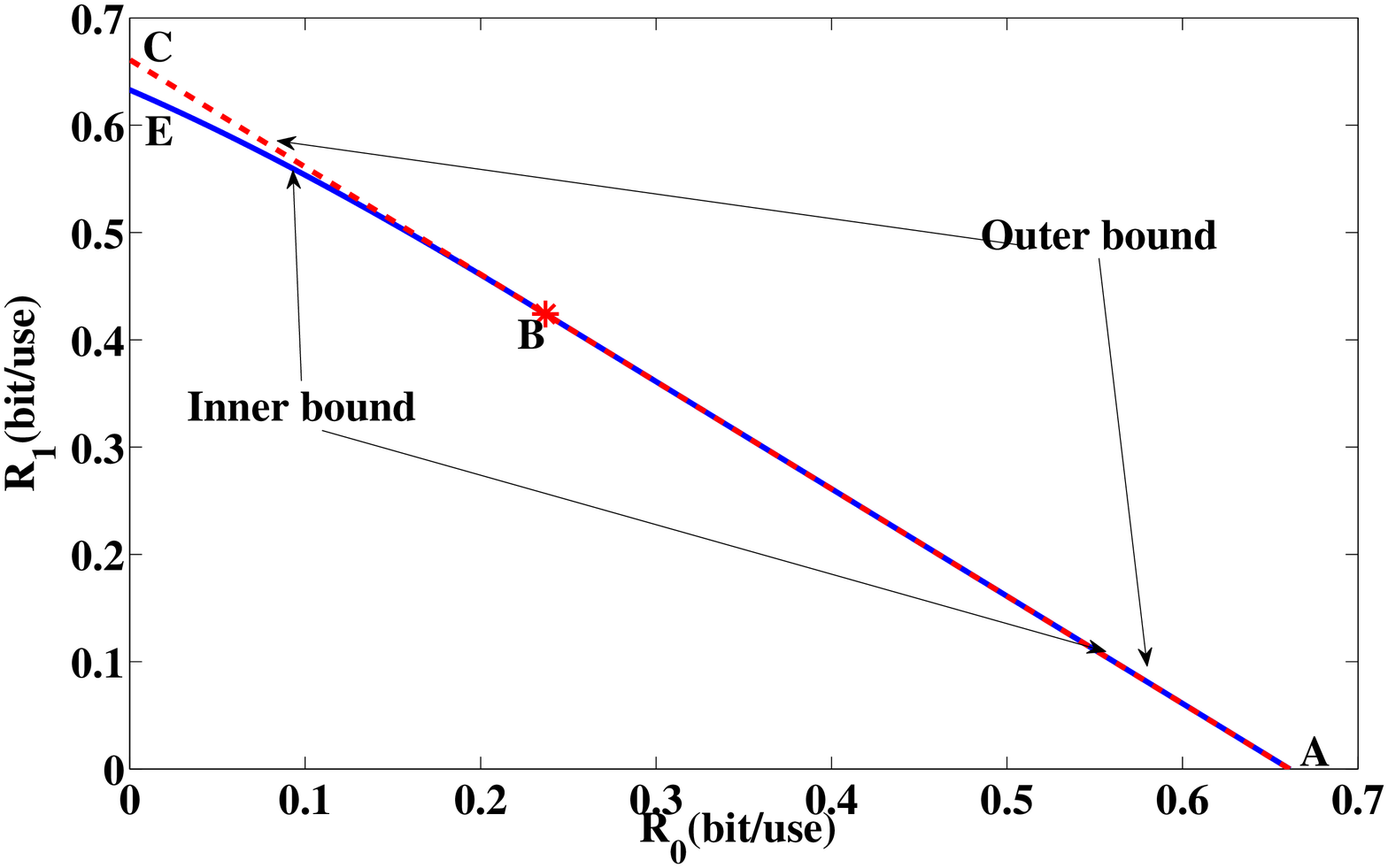}&
\includegraphics[width=2.8in]{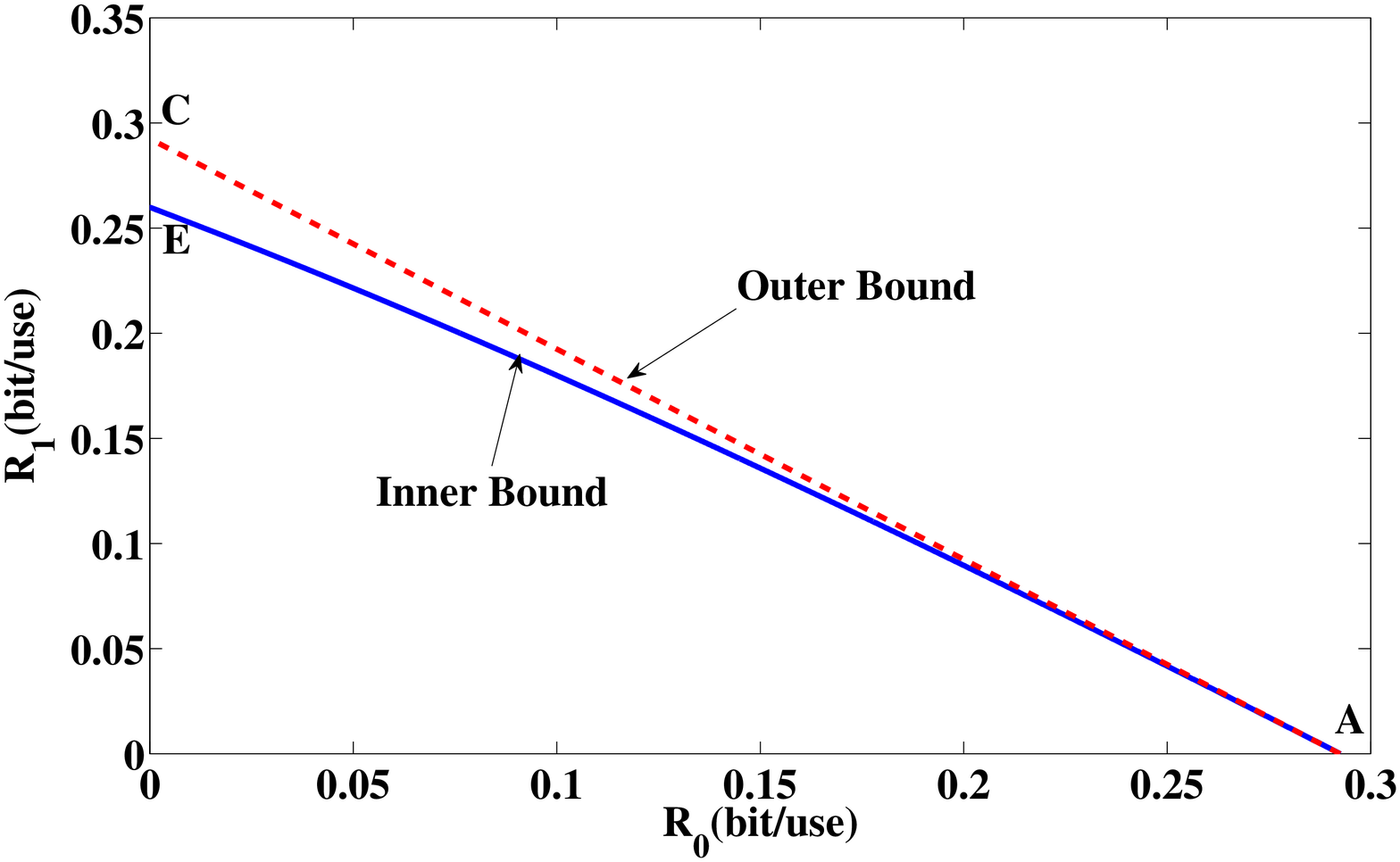} \\
\footnotesize{(a) $P_1 \geqslant 1$ with $P_0=1.5$ and $P_1=1.8$}& \footnotesize{(b) $P_1 < 1 $ with $P_0=0.5$ and $P_1=0.8$ }\\
\includegraphics[width=2.8in]{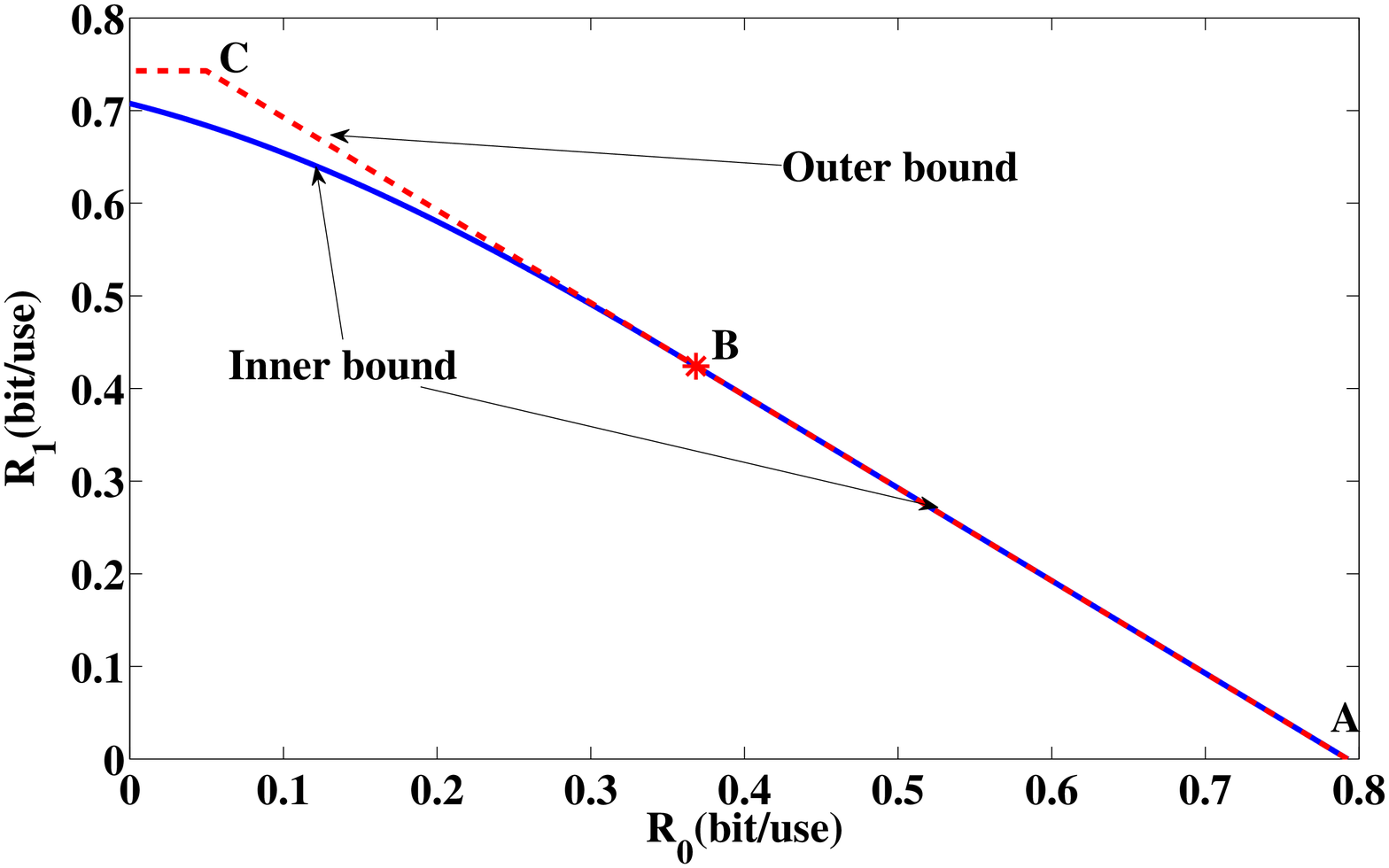}&
\includegraphics[width=2.8in]{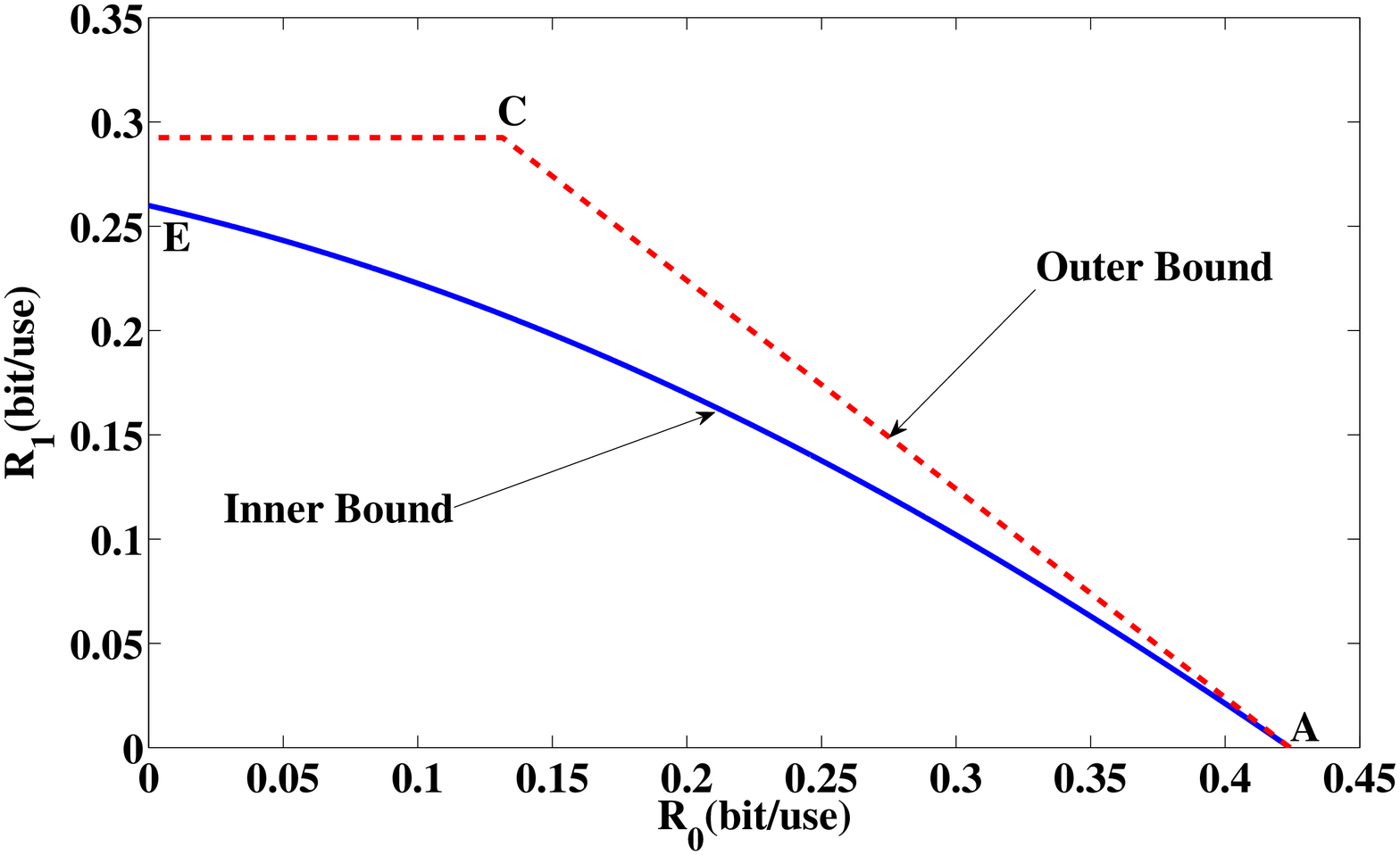}\\
\footnotesize{(c) $P_1 \geqslant 1$ with $P_0=2$ and $P_1=1.8$}&\footnotesize{(d) $P_1 < 1$ with $P_0=0.8$ and $P_1=0.5$}
\end{tabular}
\caption{Inner and outer bounds for case 2, which match partially on the boundaries}
\label{fig:SingleCase2}
\end{center}
\end{figure*}


\begin{theorem}\label{th:capacity-2}
For the Gaussian channel of model I in the regime when $Q_1\rightarrow \infty$, if $P_0-1 \leqslant P_1 < P_0 +1$ and $P_1 \geqslant 1$, the rate points $(R_0,R_1)$ on the line between $\left(\frac{1}{2}\log(1+P_0), 0\right)$ (i.e., point A in Fig.~\ref{fig:SingleCase2} (b) and (d)) and $\left(\frac{1}{2}\log( 1+\frac{ P_0- P_1+1}{P_1}),\frac{1}{2}\log P_1 \right)$ (i.e., point B in Fig.~\ref{fig:SingleCase2} (b) and (d)) are on the boundary of the capacity region.
\end{theorem}

\begin{proof}
  For the case $P_0-1 \leqslant P_1 < P_0 +1$, we also set $\alpha = \frac{2\beta P_0}{\beta P_0 +P_1 +1}$. Plug $\alpha$ into \eqref{eq:InnerGaussian-2} and we have
  \begin{subequations}
    \begin{flalign}
  &R_0 \leqslant  \frac{1}{2}\log\left( 1+\frac{\bar{\beta} P_0}{\beta P_0+1}\right)\\
  &R_1 \leqslant  \frac{1}{2}\log\left( 1+\frac{4\beta P_0P_1}{4\beta P_0 +(P_1+1-\beta P_0)^2}\right)
\end{flalign}
  \end{subequations}
  When $P_1 \geqslant 1$, by setting $\beta = \frac{P_1-1}{P_0}$, we have an achievable rate point $(R_0,R_1)$ given by $\left(\frac{1}{2}\log( 1+\frac{ P_0- P_1+1}{P_1}), \frac{1}{2}\log( P_1)\right)$, which is point $B$ in Figure \ref{fig:SingleCase2}. It can be seen that point $B$ is also on the outer bound. Obviously, Point $A$ is achievable for any set of $P_0$ and $P_1$. Thus, by time sharing the line $A-B$ is on the boundary of the capacity region.
\end{proof}


{\em Case 3: $P_1 <P_0-1$.} Similar to cases 2, the inner and outer bounds match partially over the sum rate bound, i.e., the two bounds match over the line between points A and B (see Fig.~\ref{fig:SingleCase4} (a)) if $P_1 \geqslant 1$ and match at only the point A (see Fig.~\ref{fig:SingleCase4} (b)) if $P_1 < 1$. However, differently from case 2, the inner and outer bounds also match when $R_1 = \frac{1}{2} \log(1+P_1)$ over the line between points D and E (see Fig.~\ref{fig:SingleCase4} (a) and (b)). This is because the power $P_0$ of transmitter 0 in this case is large enough to fully cancel state and signal interference so that transmitter 1 is able to reach its maximum point-to-point rate to receiver 1 without interference. Furthermore, transmitter 0 is also able to simultaneously transmit its own message at a certain positive rate as reflected by the line D-E in Fig.~\ref{fig:SingleCase4} (a) and (b). We summarize these results on the boundaries of the capacity region in the following theorem.
\begin{theorem}\label{th:capacity-4}
For the Gaussian channel of model I in the regime when $Q_1\rightarrow \infty$, if $P_1 <P_0-1$ and $P_1 \geqslant 1$, then the points on the line between $\left(\frac{1}{2}\log(1+P_0), 0\right)$ (i.e., point A in Fig.~\ref{fig:SingleCase4} (a)) and $\left(\frac{1}{2}\log( 1+\frac{ P_0- P_1+1}{P_1}),\frac{1}{2}\log P_1\right)$ (i.e., point B in Fig.~\ref{fig:SingleCase4} (a)), and the points on the line between $\left(\frac{1}{2}\log(\frac{P_0+1}{P_1+2}), \frac{1}{2}\log(1+P_1)\right)$ (i.e., point D in Fig.~\ref{fig:SingleCase4} (a)) and $\left(0,\frac{1}{2}\log(1+P_1)\right)$ (i.e., point E in Fig.~\ref{fig:SingleCase4} (a)) are on the boundary of the capacity region.

If $P_1 <P_0-1$ and $P_1 < 1$, then the point $\left(\frac{1}{2}\log(1+P_0), 0\right)$ (i.e., point A in Fig.~\ref{fig:SingleCase4} (b)) and the points on the line between $\left(\frac{1}{2}\log(\frac{P_0+1}{P_1+2}), \frac{1}{2}\log(1+P_1)\right)$ (i.e., point D in Fig.~\ref{fig:SingleCase4} (b)) and $\left(0,\frac{1}{2}\log(1+P_1)\right)$ (i.e., point E in Fig.~\ref{fig:SingleCase4} (b)) are on the boundary of the capacity region.

\end{theorem}

\begin{proof}
  For the case $P_1 <P_0-1$, the inner bound is equivalent to

Part I
\begin{subequations}
  \begin{flalign}
&0\leqslant\beta\leqslant \frac{P_1+1}{P_0}\\
       &R_0\leqslant \frac{1}{2} \log\left(1+\frac{\bar{\beta} P_0}{1+\beta P_0}\right)\\
    &R_1\leqslant \frac{1}{2} \log(1+\beta P_0)
 \end{flalign}
\end{subequations}

and Part II
\begin{subequations}
  \begin{flalign}
 &\frac{P_1+1}{P_0} \leqslant \beta \leqslant 1\\
       &R_0\leqslant \frac{1}{2} \log\left(1+\frac{\bar{\beta} P_0}{1+\beta P_0}\right)\\
    &R_1\leqslant \frac{1}{2} \log(1+P_1)\label{eq:InnerGaussian-6}
 \end{flalign}
\end{subequations}

 For Part I, when $P_1\geqslant 1$, same as in case 2, the line $A-B$ is on the boundary of the capacity region, achievable region and outer bound is shown in Figure \ref{fig:SingleCase4} (b). When $P_1 < 1$, only point $A$ on the capacity boundary is obtained, the inner bound and outer bound are indicated in Figure \ref{fig:SingleCase4}. For part II, because the bound for $\alpha$ is larger than $1$, we set $\alpha=1$, which means that both interference and state can be fully canceled at receiver 1, and the transmitting rate for $W_1$ achieves the point-to-point channel capacity(as in \eqref{eq:InnerGaussian-6}). Here, by setting $\beta = \frac{P_1+1}{P_0}$, point $D$ is achieved. Thus, the line $D-E$ is on the boundary of the capacity region. It is reasonable. Because for this case, $P_0$ is big enough to cancel the state for receiver 1 on one hand and cancel the interference caused by $X_0$ on the other hand.
\end{proof}

\begin{figure*}[hbt!]
\begin{center}
\begin{tabular}{cc}
\includegraphics[width=2.8in]{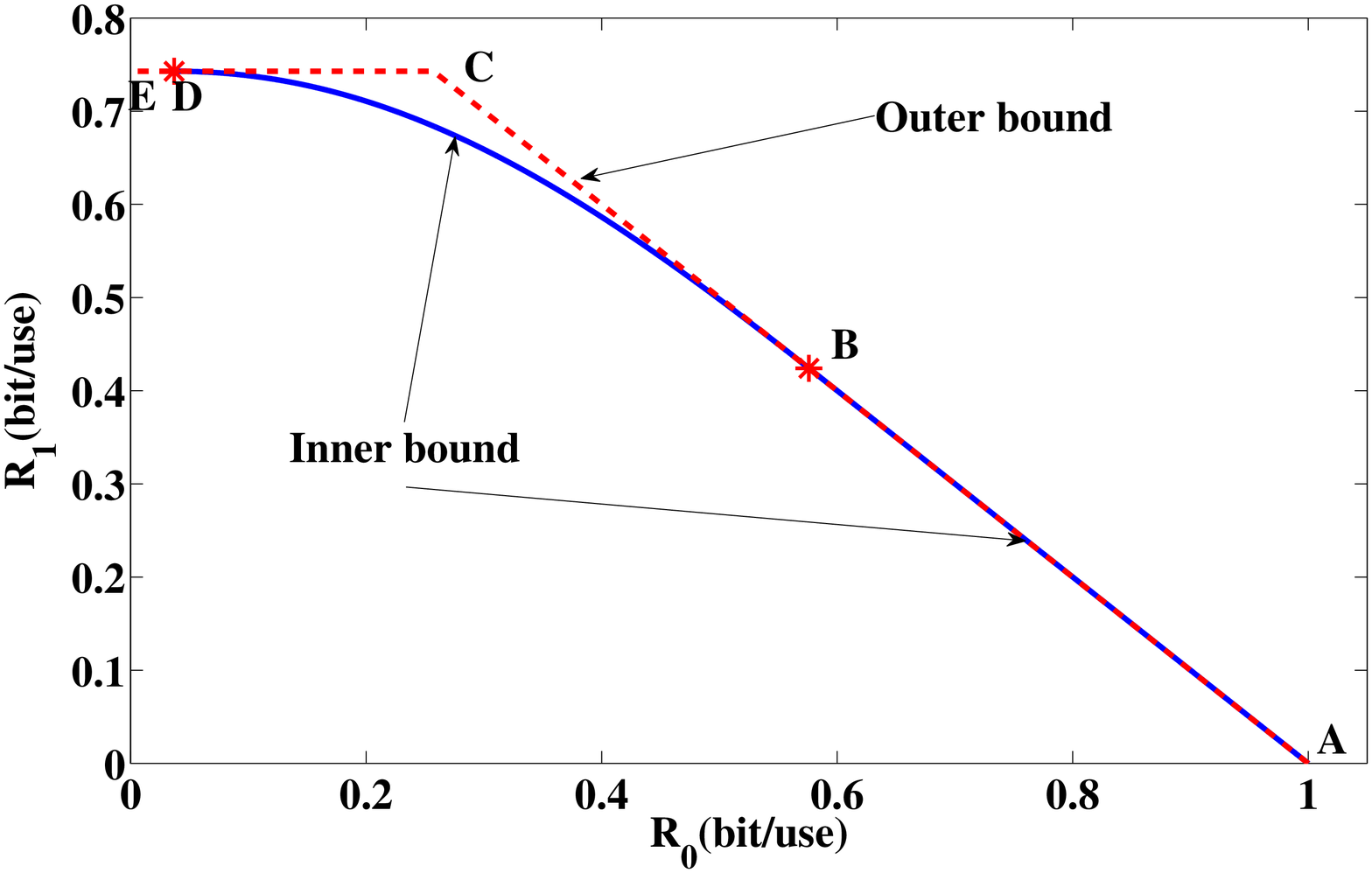}&
\includegraphics[width=2.8in]{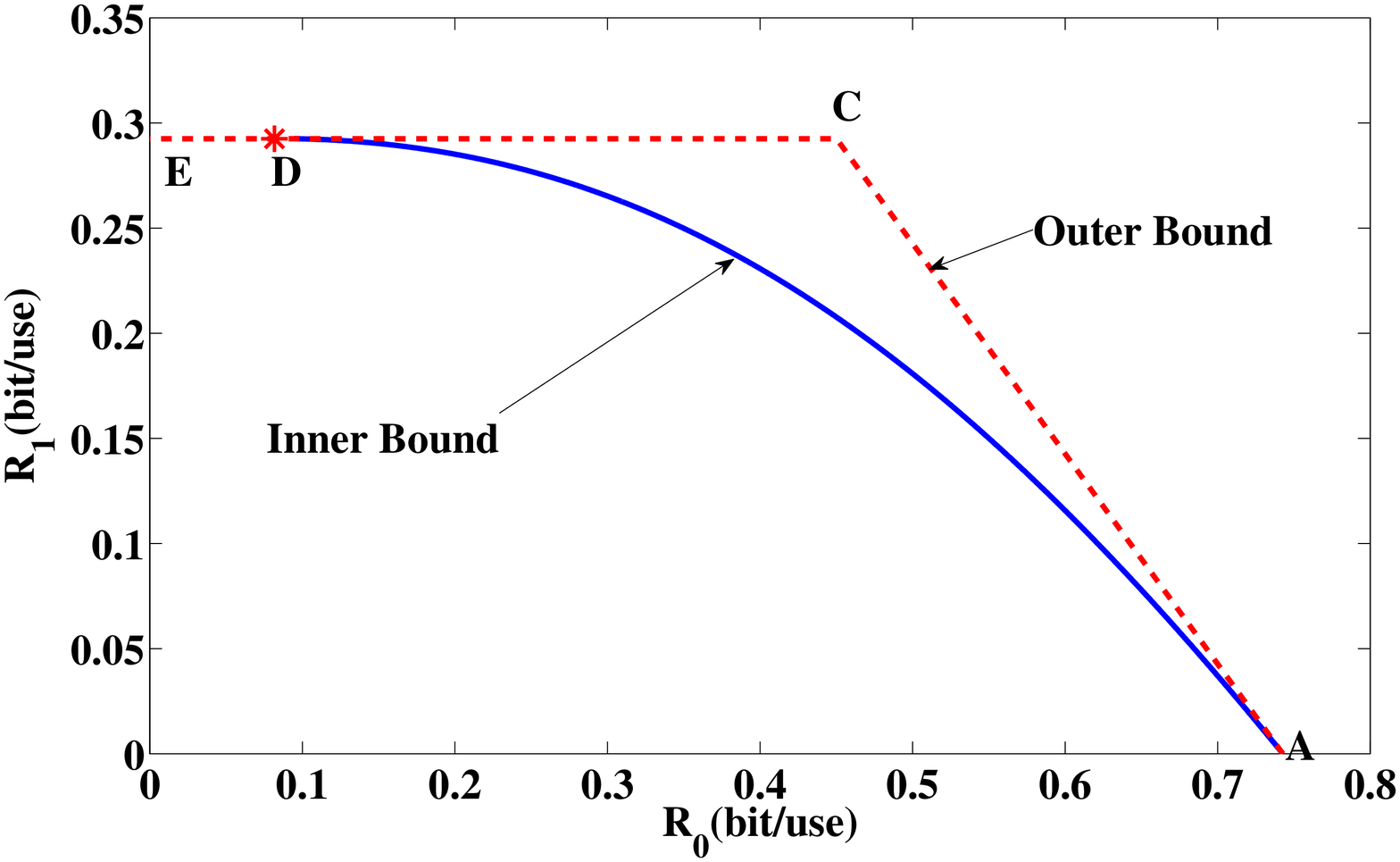} \\
\footnotesize{(a) $P_1 \geqslant 1$ with $P_0=3$ and $P_1=1.8$}& \footnotesize{(b) $P_1 < 1$ with $P_0=1.5$ and $P_1=0.5$}
\end{tabular}
\caption{Inner and outer bounds for case 3, which match partially on the boundaries}
\label{fig:SingleCase4}
\end{center}
\end{figure*}

\section{Model II: K=2}\label{sec:Same}

In this section, we consider the Gaussian channel of model II with $K=2$, but only receiver 1 interfered by an infinite state. We first study the scenario, in which the helper devotes to help two users without transmitting its own message, i.e., $W_0=\phi$. We then extend the result to the more general scenario, in which the helper also has its own message destined for the corresponding receiver in addition to helping the two users, i.e., $W_0\neq\phi$.

\subsection{Scenario with Dedicated Helper ($W_0=\phi$)}

In this subsection, we study the scenario in which only receiver 1 is corrupted by a state sequence, and the helper (without transmission of its own messages) fully assists to cancel such state interference. Here, the challenge lies in the fact that the helper needs to assist receiver 1 to remove the state interference, but such signal inevitable causes interference to receiver 2. In this subsection, we first derive an outer bound, and then derive an inner bound based on the helper using a developed layered coding scheme with one layer assisting the state-interfered receiver, and the other layer canceling the interference in order to address the challenge mentioned above. We then characterize the sum capacity under certain channel parameters and segments of the capacity boundary.

We first derive a useful outer bound for Model II.
\begin{proposition}\label{th:Outerc=0}
  For the Gaussian channel of model II with $W_0=\phi$, an outer bound on the capacity region for the regime when $Q\rightarrow \infty$ consists of rate pairs $(R_1,R_2)$ satisfying:
  \begin{subequations}
    \begin{flalign}
  R_1\leqslant & \min\left\{\frac{1}{2}\log(1+P_0),\frac{1}{2}\log(1+P_1)\right\}\label{eq:OuterGaussianc=0-1} \\
  R_2\leqslant & \frac{1}{2}\log(1+P_2)\label{eq:OuterGaussianc=0-2}\\
    R_1+R_2\leqslant & \frac{1}{2}\log(1+P_0+P_2)\label{eq:OuterGaussianc=0-3}
  \end{flalign}
  \end{subequations}
\end{proposition}

\begin{proof}
  The proof is relegated to Appendix \ref{apx:Outerc=0}
\end{proof}
We note that \eqref{eq:OuterGaussianc=0-1} represents the best single-user rate of receiver 1 with the helper dedicated to help it, as shown in Proposition \ref{th:OuterGaussianSingle}, \eqref{eq:OuterGaussianc=0-2} is the point-to-point capacity for receiver 2, and \eqref{eq:OuterGaussianc=0-3} implies that although the two transmitters communicate over the parallel channel to their corresponding receivers, due to the shared common helper, the sum rate is still subject to a certain rate limit.

We next describe our idea to design an achievable scheme. We first note that although receiver 2 is not interfered by the state, the signal that the helper sends to assist receiver 1 to deal with the state still causes unavoidable interference to receiver 2. A natural idea to optimize the transmission rate to receiver 2 is simply to keep the helper silent. In this case, without the helper's assistance, receiver 1 gets zero rate due to infinite state power. In this paper, we design a novel scheme, which enables the point-to-point channel capacity for receiver 2 and certain positive rate for receiver 1 simultaneously; i.e., the helper is able to assist receiver 1 without causing interference to receiver 2. In our achievable scheme, the signal of the helper is split into two parts, represented by $U$ and $V$ as in Proposition \ref{th:InnerSamec=0}. Here, $U$ is designed to help receiver 1 to cancel the state while treating $V$ as noise, and $V$ is designed to help receiver 2 to cancel the interference caused by $U$. Since there is no state interference at receiver 2, $U$ is decoded only at receiver 1. Based on such an achievable scheme, we obtain the following achievable region.

  \begin{proposition}\label{th:InnerSamec=0}
    For the discrete memoryless model II with $W_0=\phi$, an achievable region consists of the rate pair $(R_1,R_2)$ satisfying
    \begin{subequations}
      \begin{flalign}
                  R_{1} &\leqslant I(X_1;Y_1U)\label{eq:InnerDMCc=0-1}\\
                  R_{1} &\leqslant I(X_1U;Y_1)-I(S_1; U)\label{eq:InnerDMCc=0-2}\\
                  R_{2} &\leqslant I(X_2;Y_2V)\label{eq:InnerDMCc=0-3}\\
                  R_{2}&\leqslant I(X_2 V;Y_2)-I(V;US_1)\label{eq:InnerDMCc=0-4}
    \end{flalign}
    \end{subequations}
    for some distribution \small{$P_{S_1UVX_0X_1X_2}=P_{S_1}P_{UVX_0|S_1}P_{X_1}P_{X_2}$}, where $U$ and $V$ are auxiliary random variables.
  \end{proposition}
\begin{proof}
  The proof is detailed in Appendix \ref{apx:InnerSamec=0}.
\end{proof}

  An straight-forward subregion for the inner bound is as follows.

  \begin{corollary}\label{th:InnerDMCc=0}
  For the discrete memoryless model II with $W_0=\phi$, an achievable region consists of the rate pair $(R_1,R_2)$ satisfying
  \begin{subequations}
    \begin{flalign}
                  R_{1} &\leqslant I(X_1;Y_1|U),\\
                  R_{2} &\leqslant I(X_2;Y_2|V),
    \end{flalign}
  \end{subequations}
    for some distribution \small{$P_{S_1UVX_0X_1X_2}=P_{S_1}P_{UVX_0|S_1}P_{X_1}P_{X_2}$}, where $U$ and $V$ are auxiliary random variables such that
\begin{subequations}
  \begin{flalign}
  &I(U;Y_1)\geqslant I(U;S_1),\\
  &I(V;Y_2)\geqslant I(V;US_1).
\end{flalign}
\end{subequations}
  \end{corollary}

Following from the above achievable region, we obtain an achievable region for the Gaussian channel by setting an appropriate joint input distribution.
\begin{proposition}\label{th:InnerGaussianc=0}
  For the Gaussian channel of model II with $W_0=\phi$, an inner bound on the capacity region for the regime when $Q_1\rightarrow \infty$ consists of rate pairs $(R_1,R_2)$ satisfying:
  \begin{subequations}
  \begin{flalign}
  R_1\leqslant & \frac{1}{2}\log\left(1+\frac{P_1}{(1-\frac{1}{\alpha})^2P_{00}+P_{01}+1}\right)\label{eq:InnerGaussianc=02-1}\\
  R_2\leqslant & \frac{1}{2}\log\left(1+\frac{P_2}{1+\frac{(\beta-1)^2P_{01}P_{00}}{P_{01}+\beta^2P_{00}}}\right)\label{eq:InnerGaussianc=02-2}
  \end{flalign}
  \end{subequations}
  where $P_{00}+P_{01}\leqslant P_0$, $P_{00}, P_{01}\geqslant 0$, $0\leqslant \alpha\leqslant \frac{2P_{00}}{1+P_0+P_1}$, and $P_{01}^2+2\beta P_{00}P_{01}\geqslant \beta^2P_{00}(P_{01}+P_2+1)$.
\end{proposition}
\begin{proof}
The achievability follows from Proposition \ref{th:InnerDMCc=0} by choosing jointly Gaussian distribution for random variables as follows:
  \begin{flalign}
  U&=X_{00}+\alpha S_1,\;V=X_{01}+\beta X_{00}\nn\\
  X_{00}&\sim \mathcal{N}(0,P_{00}),\;X_{01}\sim \mathcal{N}(0,P_{01})\nn\\
  X_{1}&\sim \mathcal{N}(0,P_{1}), \;X_{2}\sim \mathcal{N}(0,P_{2})\nn
\end{flalign}
where $X_{00}$, $X_{01}$, $X_{1}$, $X_{2}$ and $S_1$ are independent.
\end{proof}

\begin{figure}[hbt!]
        \centering
        \begin{subfigure}[b]{0.48\textwidth}
                \includegraphics[width=3.4in]{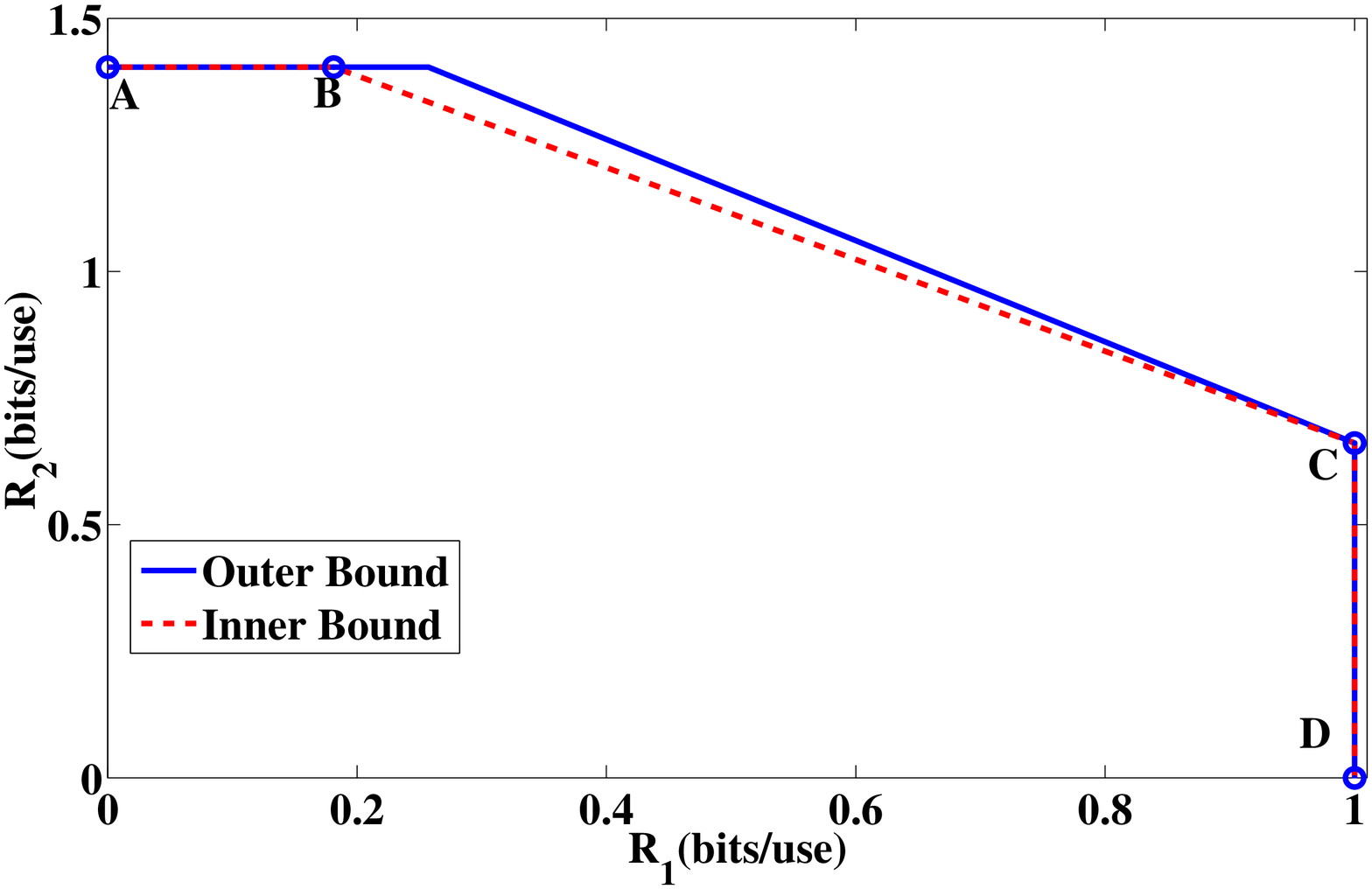}
                \caption{$P_1>P_0+1$}
                \label{fig:c=0_1}
        \end{subfigure}%
        \begin{subfigure}[b]{0.48\textwidth}
                \includegraphics[width=3.4in]{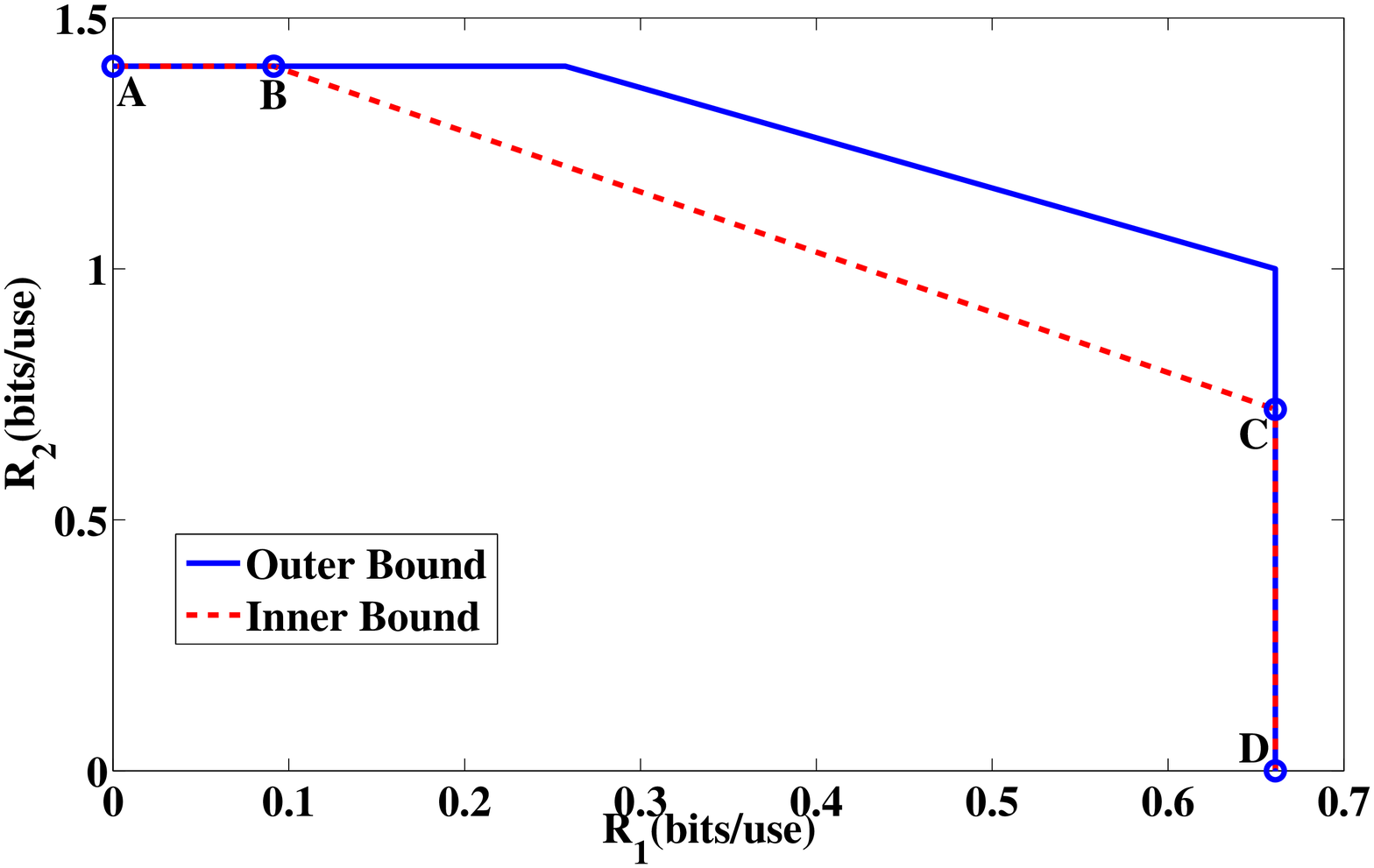}
                \caption{$P_1\leqslant P_0-1$}
                \label{fig:c=0_2}
        \end{subfigure}
        \caption{An illustration of the segments of the capacity boundary for the Gaussian channel of model II}\label{fig:c=0}
\end{figure}

Comparing the inner and outer bounds given in Propositions \ref{th:InnerGaussianc=0} and \ref{th:Outerc=0}, respectively, we characterize two segments of the boundary of the capacity region, over which the two bounds meet.
\begin{theorem}\label{th:capacityc=0}
  For the Gaussian channel of model II with $W_0=\phi$, in the regime with $Q_1\rightarrow \infty$, the line $A$-$B$ (see Figure \ref{fig:c=0}) is on the boundary of the capacity region. More specifically, if $\frac{1}{2}(1+P_0+P_1)\geqslant \frac{P_0^2}{P_0+P_2+1}$, the line $A$-$B$ is characterized as
  \small
  \begin{subequations}
    \begin{flalign}
      Point\; A:&\left(0,\frac{1}{2}\log(1+P_2)\right)\nn\\
      Point\; B:&\left(\frac{1}{2}\log \Big(1+ \frac{4P_1P_0^2}{(1+P_0+P_1)^2(1+P_0+P_2)-4P_1P_0^2}\Big),\frac{1}{2}\log(1+P_2)\right)\nn
    \end{flalign}
  \end{subequations}
  \normalsize
If $\frac{1}{2}(1+P_0+P_1) < \frac{P_0^2}{P_0+P_2+1}$, the line $A$-$B$ is characterized as
\small
  \begin{subequations}
    \begin{flalign}
      Point\; A:&\left(0,\frac{1}{2}\log(1+P_2)\right)\nn\\
      Point\; B:&\left(\frac{1}{2}\log \Big(1+\frac{P_1(P_0+P_2+1)}{P_0+(P_0+1)(P_2+1)}\Big),\frac{1}{2}\log(1+P_2)\right)\nn
    \end{flalign}
  \end{subequations}
  \normalsize
  Furthermore, the line $C$-$D$ (see Figure \ref{fig:c=0}) is also on the boundary of the capacity region. If $P_1> P_0+1$, the line $C$-$D$ is characterized as
  \small
  \begin{subequations}
    \begin{flalign}
      Point\; C:&\left(\frac{1}{2}\log(1+P_0),\frac{1}{2}\log\Big(1+\frac{P_2}{P_0+1}\Big)\right)\nn\\
      Point\; D:&\left(\frac{1}{2}\log(1+P_0),0\right)\nn
    \end{flalign}
  \end{subequations}
  \normalsize
  as illustrated in Figure \ref{fig:c=0_1}.

  If $P_1\leqslant P_0-1$, the line $C$-$D$ is characterized as
    \small
  \begin{subequations}
    \begin{flalign}
      Point\; C:&\left(\frac{1}{2}\log(1+P_1),\frac{1}{2}\log\Big(1+\frac{P_2}{P_1+2}\Big)\right)\nn\\
      Point\; D:&\left(\frac{1}{2}\log(1+P_1),0\right)\nn
    \end{flalign}
  \end{subequations}
  \normalsize
 as illustrated in Figure \ref{fig:c=0_2}.
\end{theorem}

\begin{proof}[Proof of Theorem \ref{th:capacityc=0}]
We first show that the line $A$-$B$ is achievable. The point A is achievable by keeping the helper silent. To show that the point B is achievable, we set $\beta=1$ in Proposition \ref{th:InnerGaussianc=0}, and hence the achievable rate $R_2$ in \eqref{eq:InnerGaussianc=02-2} reaches the point-to-point channel capacity, and the condition $P_{01}^2+2\beta P_{00}P_{01}\geqslant \beta^2P_{00}(P_{01}+P_2+1)$ becomes $P_{00}\leqslant \frac{P_0^2}{P_0+P_2+1}$.

  When $\frac{1}{2}(1+P_0+P_1)\geqslant \frac{P_0^2}{P_0+P_2+1}$, we have $\alpha \leqslant \frac{2P_{00}}{1+P_0+P_1} \leqslant \frac{2P_0^2}{(1+P_0+P_1)(P_0+P_2+1)} \leqslant 1$. Thus, setting $\alpha = \frac{2P_{00}}{1+P_0+P_1}$ and $P_{00}=\frac{P_0^2}{P_0+P_2+1}$, the point $\left(\frac{1}{2}\log (\frac{(1+P_0+P_1)^2(1+P_0+P_2)}{(1+P_0+P_1)^2(1+P_0+P_2)-4P_0^2}),\,\frac{1}{2}\log(1+P_2)\right)$ is achieved.

  When $\frac{1}{2}(1+P_0+P_1)\leqslant \frac{P_0^2}{P_0+P_2+1}$, we have $\alpha \leqslant \frac{2P_{00}}{1+P_0+P_1} \leqslant \frac{2P_0^2}{(1+P_0+P_1)(P_0+P_2+1)} \geqslant 1$. By setting $\alpha =1 $ and $P_{00}=\frac{P_0^2}{P_0+P_2+1}$, the point $(\frac{1}{2}\log \left(1+\frac{P_1(P_0+P_2+1)}{P_0+(P_0+1)(P_2+1)}),\,\frac{1}{2}\log(1+P_2)\right)$ is achieved.

  The line $C$-$D$ on the capacity boundary corresponds with the case when the helper is dedicated to assist receiver 1, and receiver 2 treats $X_0$ as noise, i.e., $\beta=0$ in Proposition \ref{th:InnerGaussianc=0}.

  According to the result in Theorem \ref{th:capacity-1} and Theorem \ref{th:capacity-2}, the only possible cases when the achievable rate for $R_1$ matches the outer bound is when $P_1\leqslant P_0-1$ and $P_1\geqslant P_0+1$.

  When $P_1\leqslant P_0-1$, by setting $\alpha=1$ and $P_{00}=\tilde{P}_0=P_1+1$, the rate point $\left(\frac{1}{2}\log(1+P_1),\frac{1}{2}\log(1+\frac{P_2}{P_1+2})\right)$ is achieved, which is point $C$ in Figure \ref{fig:c=0} (a). Particularly, the actual power used for the helper is $P_1+1$ rather than $P_0$, because larger $P_0$ does not better help decoding $R_1$ and increase the interference to receiver 2. It is obvious that point $D$ is achievable and points $C$ and $D$ are also on the outer bound. Hence, the points on the line $C$-$D$ is on the capacity boundary.

  When $P_1\geqslant P_0+1$, by setting the actual power of transmitter 1 as $\tilde{P}_1=P_0+1$, $P_{00}=P_0$ and $\alpha=\frac{P_0}{1+P_0}$, the rate point $\left(\frac{1}{2}\log(1+P_0),\frac{1}{2}\log(1+\frac{P_2}{1+P_0})\right)$ is achieved, which is point C in Figure \ref{fig:c=0} (b). This point also achieves the sum capacity for the channel in model II. It is obvious that point $D$ is achievable and points $C$ and $D$ are also on the outer bound. Hence, the points on the line $C$-$D$ is on the capacity boundary.
\end{proof}

The capacity result for the line $A$-$B$ in Theorem \ref{th:capacityc=0} indicates that our coding scheme effectively enables the helper to assist receiver 1 without causing interference to receiver 2. Hence, $R_2$ achieves the corresponding point-to-point channel capacity, while transmitter 1 and receiver 1 communicate at a certain rate $R_1$ with the assistance of the helper.

The capacity result for the line $C$-$D$ in Theorem \ref{th:capacityc=0} can be achieved based on a scheme, in which the helper assists receiver 1 to deal with the state and receiver 2 treats the helper's signal as noise. Such a scheme is guaranteed to be the best by the outer bound if receiver 1's rate is maximized.

Theorem \ref{th:capacityc=0} implies that if $P_1\geqslant P_0+1$, the sum capacity is achieved by the point $C$  as illustrated in Figure \ref{fig:c=0_1}.

\begin{corollary}\label{th:capacityc=0_2}
  For the Gaussian channel of model II with $W_0=\phi$, in the regime with $Q_1\rightarrow \infty$, if $P_1\geqslant P_0+1$, the sum capacity is given by $\frac{1}{2}\log(1+P_0+P_2)$.
\end{corollary}

\subsection{Extension: $W_0\neq\phi$}
In this subsection, we study the scenario, in which the helper also has its own message to transmit in addition to assisting the state-corrupted receivers, i.e., $W_0\neq\phi$. The results we present below extend those in the preceding subsection for the scenario with $W_0=\phi$. The proof techniques are similar and hence are omitted.

We first provide an outer bound for the Gaussian channel, which is generalization of Propositions \ref{th:OuterGaussianSingle} and \ref{th:Outerc=0}.
\begin{proposition}
  For the Gaussian channel of model II, an outer bound on the capacity region for the regime when $Q_1\rightarrow \infty$ consists of rate pairs $(R_1,R_2)$ satisfying:
  \begin{subequations}
    \begin{flalign}
    R_0\leqslant & \frac{1}{2}\log(1+P_0)\label{eq:OuterGaussianc=0W0-0}\\
  R_1\leqslant & \min\left\{\frac{1}{2}\log(1+P_0),\frac{1}{2}\log(1+P_1)\right\}\label{eq:OuterGaussianc=0W0-1} \\
  R_2\leqslant & \frac{1}{2}\log(1+P_2)\label{eq:OuterGaussianc=0W0-2}\\
    R_0+R_1+R_2\leqslant & \frac{1}{2}\log(1+P_0+P_2)\label{eq:OuterGaussianc=0W0-3}
  \end{flalign}
  \end{subequations}
\end{proposition}

By following similar steps as in Proposition \ref{th:InnerSamec=0} and Corollary \ref{th:InnerDMCc=0}, we derive an achievable region for the discrete memoryless model II as follows.

\begin{proposition}\label{th:InnerDMCc=0W0}
  For the discrete memoryless model II, an achievable region consists of the rate tuple $(R_0, R_1,R_2)$ satisfying
    \begin{subequations}
      \begin{flalign}
                  R_{0} &\leqslant I(X_{00};Y_0)\label{eq:InnerDMCc=0W0-0}\\
                  R_{1} &\leqslant I(X_1;Y_1|U)\label{eq:InnerDMCcW0=0-1}\\
                  R_{2} &\leqslant I(X_2;Y_2|V)\label{eq:InnerDMCcW0=0-3}
    \end{flalign}
    \end{subequations}
    for some distribution \small{$P_{S_1UVX_0X_1X_2}=P_{X_{00}}P_{S_1}P_{UVX_0|S_1X_{00}}P_{X_1}P_{X_2}$}, where $U$ and $V$ are auxiliary random variables that satisfy
\begin{subequations}
  \begin{flalign}
  &I(U;Y_1)\geqslant I(S_1X_{00}; U)\\
  &I(V;Y_2)\geqslant I(V;US_1X_{00})
\end{flalign}
\end{subequations}
\end{proposition}

Following from the above achievable region, we obtain an achievable region for the Gaussian channel by setting an appropriate joint distribution.

\begin{proposition}\label{th:InnerGaussianc=0W0}
  For Gaussian channel of model II, an inner bound on the capacity region for the regime when $Q_1\rightarrow \infty$ consists of rate tuples $(R_0, R_1,R_2)$ satisfying:
  \begin{subequations}
  \begin{flalign}
  R_0\leqslant & \frac{1}{2}\log\left(1+\frac{P_{00}}{P_{01}+P_{02}+1}\right)\label{eq:InnerGaussianc=0W02-0}\\
  R_1\leqslant & \frac{1}{2}\log\left(1+\frac{P_1}{(1-\frac{1}{\alpha})^2P_{01}+P_{02}+1}\right)\label{eq:InnerGaussianc=0W02-1}\\
  R_2\leqslant & \frac{1}{2}\log\left(1+\frac{P_2}{1+\frac{(\beta-1)^2P_{02}(P_{00}+P_{01})}{P_{02}+\beta^2(P_{00}+P_{01})}}\right)\label{eq:InnerGaussianc=0W02-2}
  \end{flalign}
  \end{subequations}
  where $P_{00}+P_{01}+P_{02}\leqslant P_0$, $P_{00}, P_{01}, P_{02}\geqslant 0$, $0\leqslant \alpha\leqslant \frac{2P_{01}}{1+P_{01}+P_{02}+P_1}$, and $P_{02}^2+2\beta P_{02}(P_{01}+P_{00})\geqslant \beta^2(P_{00}+P_{01})(P_{02}+P_2+1)$.
\end{proposition}

\begin{proof}
The achievability follows from Proposition \ref{th:InnerDMCc=0W0} by choosing jointly Gaussian distribution for random variables as follows:
  \begin{flalign}
  U&=X_{01}+\alpha (S_1+X_{00}),\;V=X_{02}+\beta (X_{00}+X_{01})\nn\\
  X_{00}&\sim \mathcal{N}(0,P_{00}),\;X_{01}\sim \mathcal{N}(0,P_{01}),\;X_{02}\sim \mathcal{N}(0,P_{02})\nn\\
  X_{1}&\sim \mathcal{N}(0,P_{1}), \;X_{2}\sim \mathcal{N}(0,P_{2})\nn
\end{flalign}
where $X_{00}$, $X_{01}$, $X_{0w}$, $X_{1}$, $X_{2}$ and $S_1$ are independent.
\end{proof}

By comparing the inner and outer bound, we reach similar conclusion as in Theorem \ref{th:capacityc=0}.

\begin{theorem}
  For the Gaussian channel of model II, in the regime with $Q_1\rightarrow \infty$, for certain rate $R_0$, i.e., certain $P_{00}$, if $\frac{1}{2}(1+P_1+P_0-P_{00})>\frac{P_0^2}{P_0+P_2+1}-P_{00}$, the line $A$-$B$ is on the boundary of the capacity region, for
  \begin{subequations}
    \begin{flalign}
      Point\; A:&\left(\frac{1}{2}\log\Big(1+\frac{P_{00}}{P_0-P_{00}+1}\Big),\, 0,\,\frac{1}{2}\log(1+P_2)\right)\nn\\
      Point\; B:&\left(\frac{1}{2}\log\Big(1+\frac{P_{00}}{P_0-P_{00}+1}\Big), \,\frac{1}{2}\log \Big(1+ \frac{P_1}{\frac{(1+P_0-P_{00}+P_1)^2(1+P_0+P_2)}{4(P_0^2-P_{00}(P_0+P2+1))}-P_1}\Big),\,\frac{1}{2}\log(1+P_2)\right)\nn
    \end{flalign}
  \end{subequations}
  \normalsize
If $\frac{1}{2}(1+P_1+P_0-P_{00})\leqslant\frac{P_0^2}{P_0+P_2+1}-P_{00}$, the line $A$-$B$ is characterized as
  \begin{subequations}
    \begin{flalign}
      Point\; A:&\left(\frac{1}{2}\log\Big(1+\frac{P_{00}}{P_0-P_{00}+1}\Big),\, 0,\,\frac{1}{2}\log(1+P_2)\right)\nn\\
      Point\; B:&\left(\frac{1}{2}\log\Big(1+\frac{P_{00}}{P_0-P_{00}+1}\Big), \,\frac{1}{2}\log \Big(1+\frac{P_1(P_0+P_2+1)}{P_0+(P_0+1)(P_2+1)}\Big),\,\frac{1}{2}\log(1+P_2)\right)\nn
    \end{flalign}
  \end{subequations}
  Furthermore, if $P_1> P_0-P_{00}+1$, the line $C$-$D$ is also on the boundary of the capacity region for the points
  \begin{subequations}
    \begin{flalign}
      Point\; C:&\left(\frac{1}{2}\log\Big(1+\frac{P_{00}}{P_0-P_{00}+1}\Big),\,\frac{1}{2}\log(1+P_{0}-P_{00}),\, \frac{1}{2}\log\Big(1+\frac{P_2}{P_0+1}\Big)\right)\nn\\
      Point\; D:&\left(\frac{1}{2}\log\Big(1+\frac{P_{00}}{P_0-P_{00}+1}\Big),\, \frac{1}{2}\log(1+P_{0}-P_{00}),\,0\right)\nn
    \end{flalign}
  \end{subequations}

  If $P_1\leqslant P_0-P_{00}-1$, the line $C$-$D$ is characterized as
    \small
  \begin{subequations}
    \begin{flalign}
      Point\; C:&\left(\frac{1}{2}\log\Big(1+\frac{P_{00}}{P_0-P_{00}+1}\Big),\, \frac{1}{2}\log(1+P_1),\,\frac{1}{2}\log\Big(1+\frac{P_2}{P_1+2}\Big)\right)\nn\\
      Point\; D:&\left(\frac{1}{2}\log\Big(1+\frac{P_{00}}{P_0-P_{00}+1}\Big),\, \frac{1}{2}\log(1+P_1),\,0\right)\nn
    \end{flalign}
  \end{subequations}
\end{theorem}

\section{Model III: General K}\label{sec:Independent}

In this section, we consider the Gaussian channel of model III with $K\ge 2$, in which there are multiple receivers with each interfered by an independent state. We first study the scenario, in which the helper dedicates to help two users without transmitting its own message, i.e., $W_0=\phi$. We then extend the result to the more general scenario, in which the helper also has its own message destined for the corresponding receiver in addition to helping the two users, i.e., $W_0\neq\phi$.

\subsection{Scenario with Dedicated Helper ($W_0=\phi$)}\label{sec:modelIII1}

In this subsection, we study the scenario in which there are multiple receivers with each corrupted by an independently distributed state sequence, and the common helper (without transmission of its own messages) fully assists to cancel such state interference for all receivers. We note that this model is more general than model I, but does not include model II as a special case, because model II has one receiver that is not corrupted by state, but each receiver (excluding the helper) in model III is corrupted by an infinitely powered state sequence. Hence for model III, the challenge lies in the fact that the helper needs to assist multiple receivers to cancel interference caused by independent states. In this subsection, we first derive an outer bound on the capacity region, and then derive an inner bound based on a time-sharing scheme for the helper. Somewhat interestingly, comparing the inner and outer bounds concludes that the time-sharing scheme achieves the sum capacity under certain channel parameters, and we hence characterize certain segments of the boundary of the capacity region corresponding to the sum capacity under these channel parameters.

In the following, we first study the simple case with $K=2$, which is more instructional. We then generalize our results to the case with $K \ge 2$. We first derive the outer bound on the capacity region.
\begin{proposition}\label{th:OuterIndepend2}
  For the Gaussian channel of model III with $W_0=\phi$ and $K=2$, an outer bound on the capacity region for the regime when $Q_1, Q_2\rightarrow \infty$ consists of rate pairs $(R_1, R_2)$ satisfying:
    \begin{flalign}
        R_1 & \leqslant \frac{1}{2}\log(1+P_1)\label{eq:OuterQ2-1}\\
  R_2 & \leqslant \frac{1}{2}\log(1+P_2)\label{eq:OuterQ2-2}\\
  R_1 +R_2 & \leqslant \frac{1}{2}\log(1+P_0)\label{eq:OuterQ2-3}
    \end{flalign}
\end{proposition}
\begin{proof}
The proof is detailed in Appendix \ref{apx:OuterIndepend2}.
\end{proof}

We note that although the two transmitters transmit over a parallel channel, the above outer bound suggests that their sum rate is still subject to a certain constraint determined by the helper's power. This implies that it is not possible for one common helper to cancel the two independent high-power states simultaneously (i.e., using the common resource). This fact also suggests that a time-sharing scheme, in which the helper alternatively assists each receiver, can be desirable to achieve the sum rate upper bound (i.e., to achieve the sum capacity).

We hence design the following time-sharing achievable scheme. The helper splits its transmission duration into two time slots with the fraction $\gamma$ of the total time duration for assisting receiver 1 and the fraction $1-\gamma$ for assisting receiver 2. Each transmitter transmits only during the time slot that it is assisted by the helper, and keeps silent while the helper assisting the other transmitter. We note that the power constraints for transmitters 1 and 2 in their corresponding transmission time slots are $\frac{P_1}{\gamma}$ and $\frac{P_2}{1-\gamma}$, respectively.

Now at each transmission slot, the channel consists of one transceiver pair with the receiver corrupted by a high-power state, and one helper that assists the receiver to cancel the state interference. Such a model is a special case of the model studied in Section \ref{sec:ResultSingle} (with the state-informed transmitter not having its own message). We rewrite the achievable rate for the model in Section \ref{sec:ResultSingle} with power constraints $P$ and $P_0$ respectively at the transmitter and the helper:
\begin{flalign}
  R&(P,P_0) \nn \\
& :=\begin{cases}
\frac{1}{2}\log(1+P_0), &  P\geqslant P_0+1\\
\frac{1}{2}\log(1+\frac{4P_0P}{4P_0+(P_0-P-1)^2}), \;\; & P_0-1\leqslant P\leqslant P_0+1 \\
\frac{1}{2}\log(1+P), & P\leqslant P_0-1
\end{cases}\label{eq:InnerSingle}
\end{flalign}
By employing the time-sharing scheme between the helper assisting one receiver and the other alternatively, we obtain the following achievable region.
\begin{proposition}\label{th:InnerIndepend2}
For the Gaussian channel of model III with $W_0=\phi$ and $K=2$, in the regime with $Q_1, Q_2\rightarrow \infty$, an inner bound on the capacity region consists of rate pairs $(R_1, R_2)$ satisfying:
\begin{subequations}
  \begin{flalign}
  R_1 &\leqslant \gamma R\left(\frac{P_1}{\gamma},P_0\right)\\
  R_2 &\leqslant (1-\gamma) R\left(\frac{P_2}{1-\gamma},P_0\right)
\end{flalign}
\end{subequations}
where $0\leqslant \gamma\leqslant 1$ is the time-sharing coefficient.
\end{proposition}

We note that following from \eqref{eq:InnerSingle}, the best possible single-user rate is $\frac{1}{2}\log(1+P_0)$, which can be achieved if $P\geqslant P_0+1$. This best rate may not be possible if $P$ is not large enough. Interestingly, in a time-sharing scheme, both transmitters can simultaneously achieve the best single user rate $\frac{1}{2}\log(1+P_0)$ over their transmission fraction of time, because both of their powers get boosted over a certain fraction of time, although neither power is larger than $P_0+1$. In this way, the sum rate upper bound \eqref{eq:OuterQ2-3} can be achieved. The following theorem characterizes the sum capacity of the channel for the scenario describe above.
\begin{theorem}\label{th:IndependC}
For the Gaussian channel of model III with $K=2$ and $W_0=\phi$, in the regime with $Q_1, Q_2\rightarrow \infty$, if $P_1+P_2\geqslant P_0+1$, the sum capacity equals $\frac{1}{2}\log(1+P_0)$. The rate points that achieve the sum capacity on the boundary of the capacity region are characterized as $(R_1,R_2)=\left(\gamma R(\frac{P_1}{\gamma},P_0),(1-\gamma) R(\frac{P_2}{1-\gamma},P_0)\right)$ for $\gamma \in \left(\max(1-\frac{P_2}{P_0+1},0),\min(\frac{P_1}{P_0+1},1)\right)$.
\end{theorem}
\begin{proof}
 The proof is detailed in Appendix \ref{apx:IndependentC}.
\end{proof}

The above theorem implies the following characterization of the full capacity region under certain parameters.
\begin{corollary}\label{th:capacity_model2}
  For the Gaussian channel of model III with $W_0=\phi$, in the regime with $Q_1, Q_2\rightarrow \infty$, if $P_1,\;P_2\geqslant P_0+1$, then the capacity region consists of the rate pair $(R_1,R_2)$ satisfying $R_1+R_2\leqslant\frac{1}{2}\log(1+P_0)$.
\end{corollary}

We next provide channel examples to understand our inner and outer bounds, and in particular, Theorem \ref{th:IndependC} on the sum capacity. It can be seen that the power constraints fall into four cases, among which we consider the following three cases: case 1. $P_1\geqslant P_0, P_2\geqslant P_0$; case 2. $P_1\geqslant P_0, P_2<P_0$; and case 3. $P_1 <  P_0, P_2 < P_0$ by noting that case 4 is opposite to case 2 and is omitted due to symmetry of the two transmitters.

\begin{itemize}
\item Case 1: $P_1\geqslant P_0, P_2\geqslant P_0$

We consider an example channel with $P_0=1$, $P_1=1.8$ and $P_2=1.5$. Figure \ref{fig:Independ} (a) plots the inner and outer bounds on the capacity region. In particular, the two bounds meet over the line segment B-C, which corresponds to the rate points $(R_1,R_2)=\left(\gamma R(\frac{P_1}{\gamma},P_0),(1-\gamma) R(\frac{P_2}{1-\gamma},P_0)\right)$ for $\gamma \in \left(\max(1-\frac{P_2}{P_0+1},0),\min(\frac{P_1}{P_0+1},1)\right)$ as characterized in Theorem \ref{th:IndependC}. All these rate points achieve the sum capacity. It can also be seen that although neither transmitter achieves the best possible single-user rate, the sum capacity can be achieved due to the time-sharing scheme. We also note that, in this case, if the conditions in Corollary \ref{th:capacity_model2} are satisfied, the full capacity region is characterized.

\item Case 2: $P_1>P_0, P_2\leqslant P_0$

We consider an example channel with $P_0=2$, $P_1=2.5$ and $P_2=0.8$. Figure \ref{fig:Independ} (b) plots the inner and outer bounds on the capacity region. Similarly to case 1, the two bounds meet over the line segment B-C as characterized in Theorem \ref{th:IndependC}, and the points on such a line segment achieve the sum capacity. Differently from case 1, transmitter 2 achieves its point-to-point channel capacity indicated by the point A in Figure \ref{fig:Independ} (b). This is consistent with the single user rate provided in \eqref{eq:InnerSingle} for the case with $P_2\leqslant P_0-1$.

  \item Case 3: $P_1 < P_0, P_2 < P_0$

  We consider an example channel with $P_0=4$, $P_1=3$ and $P_2=3$. Figure \ref{fig:Independ} (c) plots the inner and outer bounds on the capacity region. The points on the line segment B-C achieve the sum capacity as characterized in Theorem \ref{th:IndependC}, and the points A and D respectively achieve the point-to-point capacity for two transceiver pairs. This is consistent with the single-user rate provided in \eqref{eq:InnerSingle} for the case with $P_1, P_2\leqslant P_0-1$.
\begin{figure*}[hbt!]
\begin{center}
\begin{tabular}{ccc}
\includegraphics[width=2in]{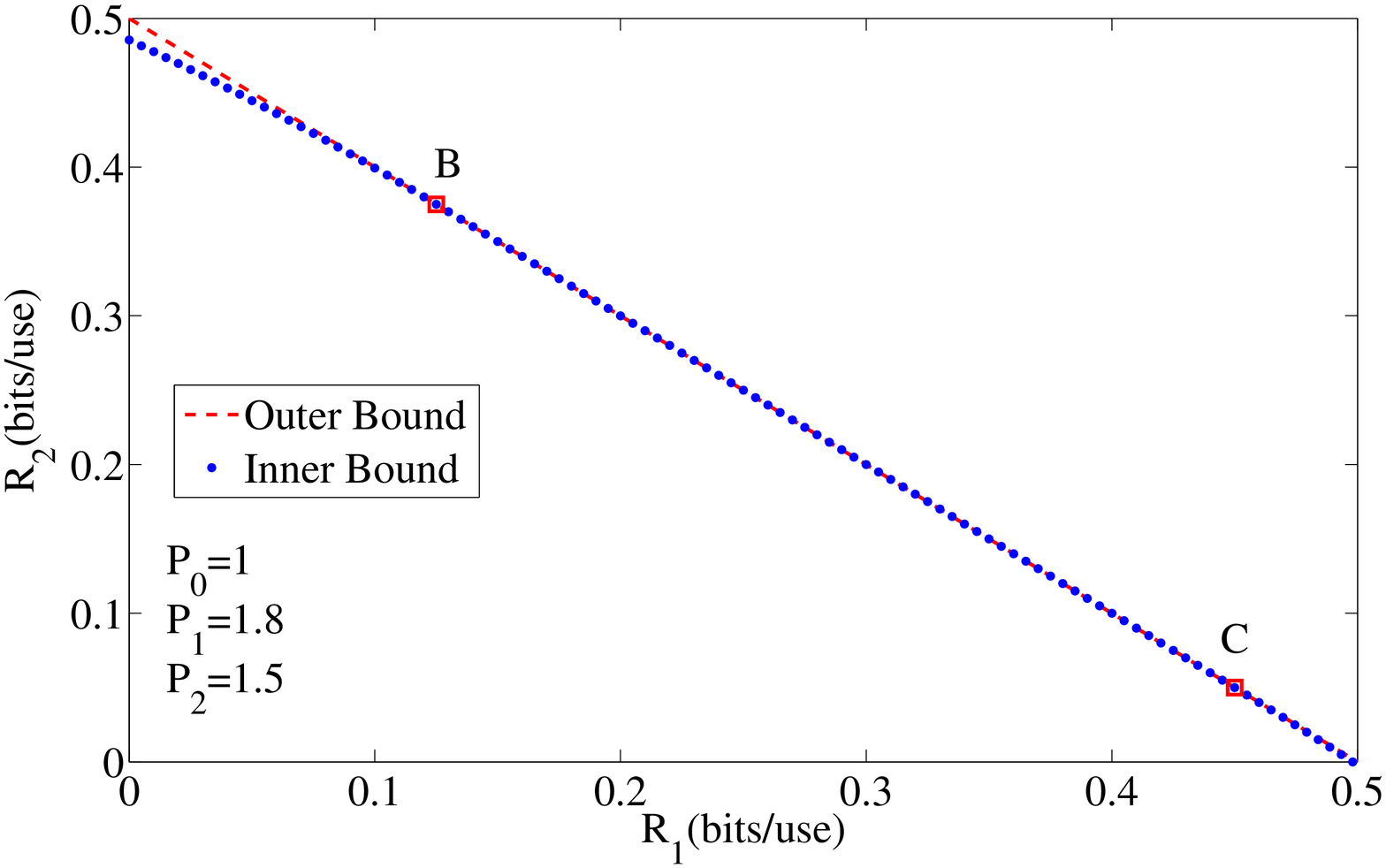} &\includegraphics[width=2in]{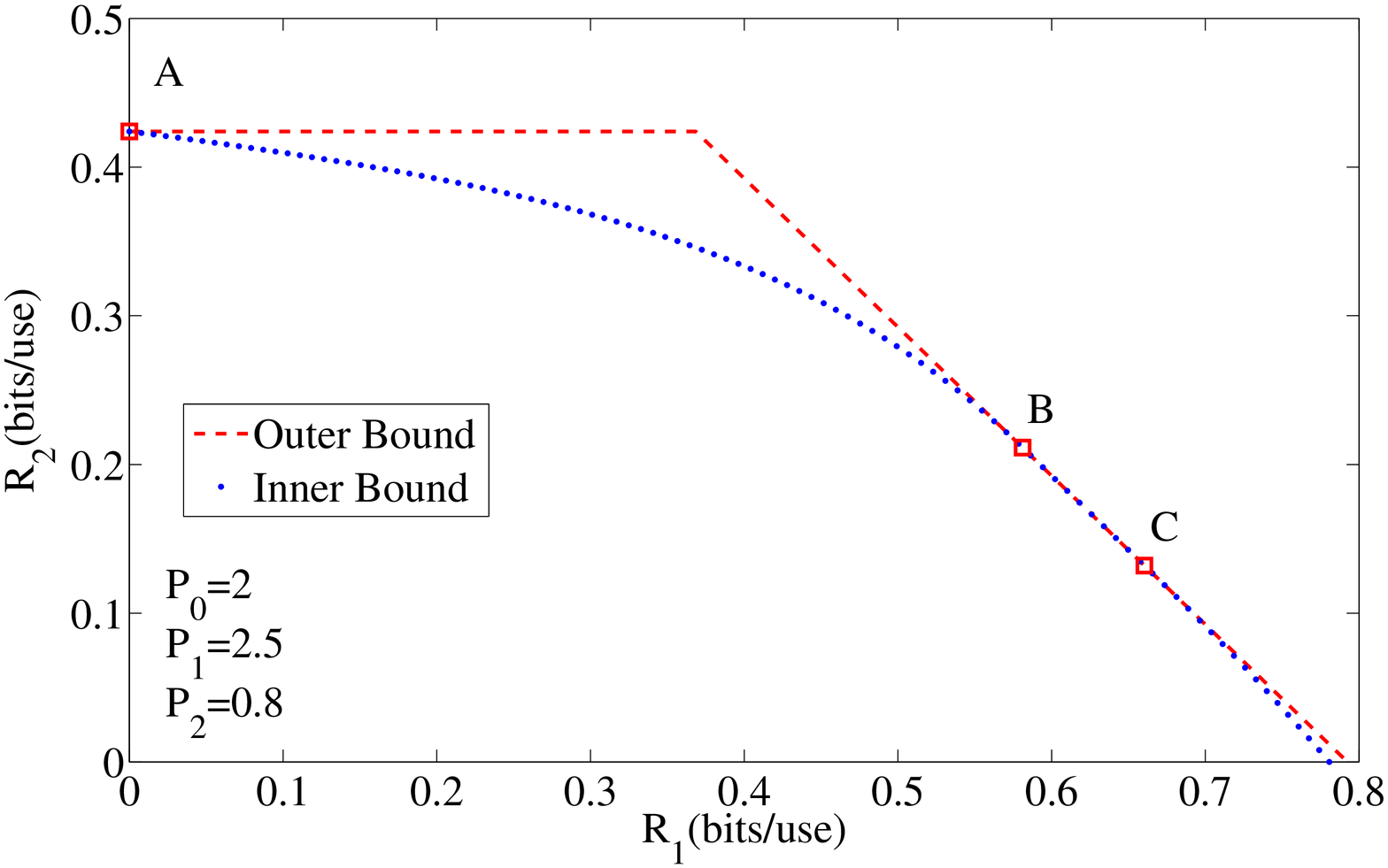} & \includegraphics[width=2in]{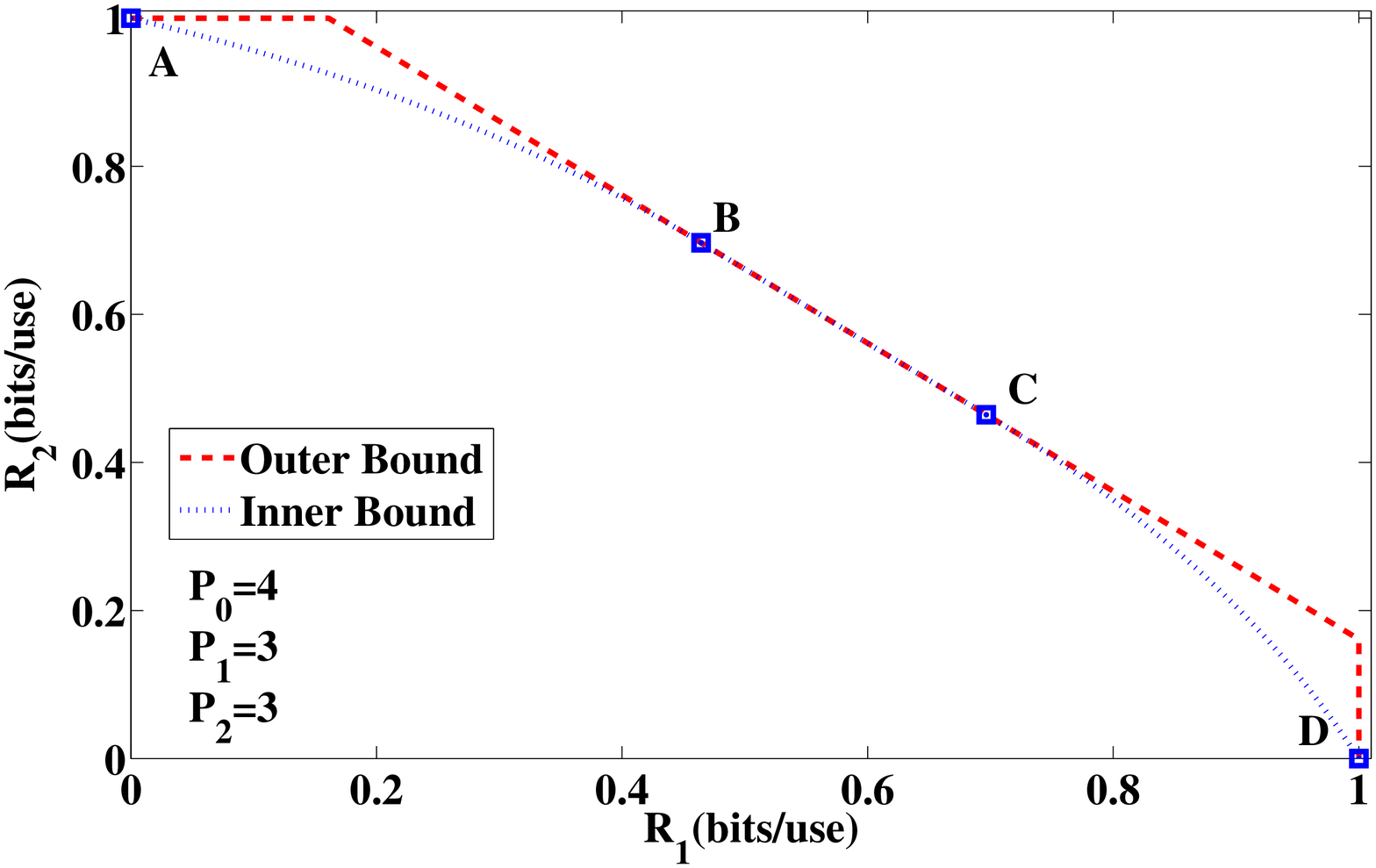} \\
\footnotesize{(a) $P_1\geqslant P_0, P_2\geqslant P_0$ }& \footnotesize{(b) $P_1\geqslant P_0, P_2< P_0$}& \footnotesize{(c) $P_1 < P_0, P_2 < P_0$}
\end{tabular}
\caption{An illustration of the partial capacity boundary for the Gaussian channel of model III}
\label{fig:Independ}
\vspace{-0.8cm}
\end{center}
\end{figure*}

\end{itemize}

We next generalize our results to the case with $K \ge 2$.

\subsection{Extension: $W_0\neq\phi$}\label{sec:Exension}
In this subsection, we study the scenario, in which the helper also has its own message to transmit in addition to assisting the state-corrupted receivers, i.e., $W_0\neq\phi$. The results we present below extend those in the preceding subsection for the scenario with $W_0=\phi$ as well as the results in Section \ref{sec:ResultSingle} for model I. The proof techniques combine those in Sections \ref{sec:modelIII1} and \ref{sec:ResultSingle}, and hence are omitted.


We first provide the outer bound as follows.
 \begin{proposition}\label{th:OuterIndependKW0}
For the Gaussian channel of model III, an outer bound on the capacity region for the regime when $Q_1,\dots, Q_K\rightarrow \infty$ consists of rate tuples $(R_0,\dots, R_K)$ satisfying:
  \begin{flalign}
  \sum_{k=0}^K& R_k\leqslant \frac{1}{2}\log(1+P_0)\nn\\
  R_k& \leqslant \frac{1}{2}\log(1+P_k)\;\;\;\;k=0,\dots,K\nn
\end{flalign}
\end{proposition}

\begin{proof}
  The proof is detailed in Appendix \ref{apx:OuterIndependKW0}
\end{proof}

By utilizing time-sharing scheme, receivers share the helper's assistance with activating time portion $\gamma_k$ for receiver k. The corresponding achievable region is as follows.
\begin{proposition}\label{th:InnerIndependKW0}
For the Gaussian channel of model III, an inner bound consists of rate tuples $(R_0,\dots, R_K)$ satisfying:
  \begin{flalign}
  R_0&\leqslant \sum_{k=1}^K \gamma_k \frac{1}{2}\log\left(1+\frac{\bar{\beta}_kP_0}{\beta_k P_0+1}\right)\nn\\
  R_k &\leqslant \gamma_k R\left(\frac{P_k}{\gamma_k},\beta_k P_0\right)\nn\\
  \sum_{k=1}^K & \gamma_k =1,\;\;\gamma_k\geqslant0\;\;\bar{\beta}_k+\beta_k=1\;\;\;\;\; k=1,\dots,K\nn
\end{flalign}
where $R(\cdot,\cdot)$ are the function defined in \eqref{eq:InnerSingle}.
\end{proposition}

By comparing the inner and outer bound, we introduce a natural generalization of Theorem \ref{th:IndependC}.
\begin{theorem}
  The sum capacity of the Gaussian channel of model III is $\sum_{k=0}^K R_k \leqslant \frac{1}{2}\log(1+P_0)$. The points on the boundary are characterized as
\begin{flalign}
  R_0&\leqslant \sum_{k=1}^K \gamma_k \frac{1}{2}\log\left(1+\frac{\bar{\beta}_kP_0}{\beta_k P_0+1}\right)\nn\\
  R_k &\leqslant \gamma_k \frac{1}{2}\log(1+\beta_k P_0)\;\;\;\;\;\; k=1,\dots,K\nn
\end{flalign}
where
\begin{flalign}
  \sum_{k=1}^K & \gamma_k=1,\;\;\gamma_k\geqslant0\nn\\
  \frac{P_k}{\gamma_k}&\geqslant \beta_k P_0+1\nn\\
  \bar{\beta}_k&+\beta_k=1\;\;\;\;\; k=1,\dots,K.\nn
\end{flalign}
\end{theorem}

\section{Conclusion}\label{sec:Conclusion}

In this paper, we proposed and studied parallel communication networks with a state-cognitive helper. We considered three models, and derived inner and outer bounds for each model. By comparing these bounds, we characterized full or certain segments of the boundary of the capacity region for some channel parameters. As we mentioned in Section \ref{sec:Introduction}, large state interference and state-cognitive helpers are well justified in practical wireless networks, and hence the achievable schemes developed here are promising to greatly improve the throughput of wireless networks. We also anticipate that the techniques that we develop in this paper will be helpful for studying various other multi-user state-dependent models with state-cognitive helpers.


\vspace{10mm}

\appendix

\noindent {\Large \textbf{Appendix}}
\section{Proof of Lemma \ref{th:InnerDMCSingle}}\label{apx:InnerDMCSingle}
We use random codes and fix the following joint distribution:
\[P_{S_1X_0'UX_0X_1Y_0Y_1}= P_{S_1}P_{X_0'}P_{U|S_1X_0'}P_{X_0|US_1X_0'}P_{X_1}P_{Y_0|X_0 }P_{Y_1|X_0X_1S_1}.\]
Let $T_\epsilon^n(P_{S_1X_0'UX_0X_1Y_0Y_1})$ denote the strongly joint $\epsilon$-typical set based on the above distribution. For a given sequence $x^n$, let $T_\epsilon^n(P_{U|X}|x^n)$ denote the set of sequences $u^n$ such that $(u^n, x^n)$ is jointly typical based on the distribution $P_{XU}$.
\begin{enumerate}{\topsep=0.ex \leftmargin=0.25in
\rightmargin=0.in \itemsep =0.in}
  \item Codebook Generation
  \begin{itemize}
    \item Generate $2^{n\tilde{R}}$ codewords $u^n(v)$ with the probability of $P_{U}$, in which $v \in [1,2^{n\tilde{R}}]$. 
    \item Generate $2^{nR_0}$ codewords $x_0^{'n}(w_0)$ with the probability of $P_{X'_0}$, in which $w_0 \in [1,2^{nR_0}]$.
    \item Generate $2^{nR_1}$ codewords $x_1^{n}(w_1)$ with the probability of $P_{X_1}$, in which $w_1 \in [1,2^{nR_1}]$.
  \end{itemize}

  \item Encoding
  \begin{itemize}
    \item Encoder 0: Given $w_0$, map $w_0$ into $x_0^{'n}(w_0)$.
                    For each $x_0^{'n}(w_0)$, select $\tilde{v}$ such that $(u^n(\tilde{v}), s_1^n, x_0^{'n}(w_0)) \in T_\epsilon^n(P_{S_1} P_{X_0'}P_{U|S_1X_0'})$.
                    If $u^n(\tilde{v})$ cannot be found, set $\tilde{v}=1$. Then map $(s_1^n, u^n(\tilde{v}), x_0^{'n}(w_0))$ into $x_0^n = f^{(n)}(x_0^{'n}(w_0), s_1^n, u^n(\tilde{v}))$. It can be shown that such $u^n(\tilde{v})$ exists with high probability for large $n$ if
                      \begin{equation}\label{eq:InnerDMCSinglePF-1}
                            \tilde{R} > I(U;S_1X_0').
                      \end{equation}

    \item Encoder 1: Given $w_1$, map $w_1$ into $x_1^n(w_1)$.
  \end{itemize}

  \item Decoding
  \begin{itemize}
    \item Decoder 0: Given $y_0^n$, find $\hat{w}_0$ such that $(x_0^{'n}(\hat{w}_0), y_0^n) \in T_\epsilon^n(P_{X'_0Y_0})$. If no or more than one $\hat{w}_0$ can be found, declare error. It can be shown that the decoding error is small for sufficient large $n$ if
               \begin{equation}\label{eq:InnerDMCSinglePF-2}
                  R_{0} \leqslant I(X_0';Y_0).
               \end{equation}
\item Decoder 1: Given $y_1^n$, find a pair $(\hat{v}, \hat{w}_{1})$ such that $(u^n(\hat{v}), x_1^n(\hat{w}_1), y_1^n) \in T_\epsilon^n(P_{UX_1Y_1})$. If no or more than one such pair can be found, then declare error. It can be shown that decoding is successful with small probability of error for sufficiently large $n$ if  the following conditions are satisfied
    \begin{flalign}
      R_{1} \leqslant & I(X_1;Y_1|U),\label{eq:InnerDMCSinglePF-3}\\
      \tilde{R} \leqslant & I(U;Y_1|X_1),\label{eq:InnerDMCSingleRe-1}\\
      R_1 + \tilde{R}\leqslant & I(UX_1;Y_1).\label{eq:InnerDMCSinglePF-4}
    \end{flalign}

  \end{itemize}
\end{enumerate}

We note that \eqref{eq:InnerDMCSingleRe-1} corresponds to the decoding error for the index $v$, which is not the message of interest. Hence, the bound \eqref{eq:InnerDMCSingleRe-1} can be removed. Hence, combining \eqref{eq:InnerDMCSinglePF-1}, \eqref{eq:InnerDMCSinglePF-2}, \eqref{eq:InnerDMCSinglePF-3}, and \eqref{eq:InnerDMCSinglePF-4} and eliminating $\tilde{R}$, and we obtain the desired achievable region as in Lemma \ref{th:InnerDMCSingle}.

\section{Proof of Proposition \ref{th:OuterGaussianSingle}}\label{apx:OuterGaussianSingle}
We first bound the single-user rate $R$ as follows.
\begin{equation}
  R_1 \leqslant \frac{1}{2} \log(1+P_1)
\end{equation}
We then bound the sum rate as follows. For the message $W_0$, based on  Fano's inequality, we have
\begin{flalign}\label{eq:OuterProof-1}
    nR_0&\leqslant I(W_0;Y_0^n)+n\epsilon_n \\\nn
    & =h(Y_0^n)-h(Y_0^n|W_0)+n\epsilon_n,\nn
  \end{flalign}
where $\epsilon_n\rightarrow 0$ as $n\rightarrow \infty$.

For the message $W_1$, based on Fano's inequality, we have
    \begin{flalign}\label{eq:OuterProof-2}
    nR_1&\leqslant I(W_1;Y_1^n)+n\epsilon_n \\\nn
    & =h(Y_1^n)-h(Y_1^n|W_1)+n\epsilon_n\\\nn
    & \leqslant h(Y_1^n)-h(Y_1^n|W_1 X_1^n)+n\epsilon_n\\\nn
    & = h(Y_1^n)-h(X_0^n+S_1^n+N_1^n)+n\epsilon_n\\\nn
    & \leqslant h(Y_1^n)-h(X_0^n+S_1^n+N_1^n|W_0 Y_0^n)+n\epsilon_n\nn
  \end{flalign}
  Summation of \eqref{eq:OuterProof-1} and \eqref{eq:OuterProof-2} yields
\begin{flalign}
    n(R_0+R_1) &\leqslant h(Y_0^n) + h(Y_1^n)-h(Y_0^n, X_0^n+S_1^n+N_1^n|W_0 )
  \end{flalign}

Since the capacity region of the channel depends on only marginal distributions of $(X_0, Y_0)$ and $(X_0,X_1,S,Y_1)$, setting $N_1=N_0$ does not change the capacity region. Thus,
\begin{flalign}
    n(R_0+R_1) &\leqslant h(Y_0^n) + h(Y_1^n)-h(S_1^n, X_0^n+N_1^n|W_0)\\\nn
    & \leqslant h(Y_0^n) + h(Y_1^n)-h(S_1^n)-h(N_1^n)\\\nn
    & \leqslant \frac{n}{2} \log(1+P_0) +\frac{n}{2} \log\left(1+\frac{P_0+P_1+1}{Q_1}\right)
  \end{flalign}

As $Q_1 \rightarrow \infty$, the second term of the above bound goes to $0$, and we have
  \begin{flalign}
    n(R_0+R_1)  \leqslant \frac{1}{2} \log(1+P_0).
  \end{flalign}

\section{Proof of Proposition \ref{th:Outerc=0}}\label{apx:Outerc=0}
The single rate bound is based on the result in Section \ref{sec:ResultSingle} for model I and the point-to-point channel capacity.

For the sum rate bound, according to Fano's inequality , we have
  \begin{flalign}
    n(R_1+R_2)\leqslant &I(W_1;Y_1^n)+I(W_2;Y_2^n)\nn\\
=&h(Y_1^n)-h(Y_1^n|W_1)+h(Y_2^n)-h(Y_2^n|W_2)\nn\\
\overset{(a)}{=}&h(Y_1^n)-h(Y_1^n|W_1 X_1^n)+h(Y_2^n)-h(Y_2^n|W_2 X_2^n)\nn\\
=&h(Y_1^n)-h(X_0^n+S_1^n+N_1^n)+h(Y_2^n)-h(X_0^n+N_2^n)\nn\\
\leqslant & h(Y_1^n)-h(X_0^n+S_1^n+N_1^n| X_0^n+N_1^n)\nn\\
    &+h(Y_2^n)-h(X_0^n+N_2^n)\nn
  \end{flalign}
where (a) follows from that $X_1^n$ is function of $W_1$, and $X_2^n$ is function of $W_2$, and they are independent from $X_0^n$, state and noise. Because the two decoders decode based on the marginal distribution only, setting $N_1^n=N_2^n$ does not influence the channel capacity, therefore,
\begin{flalign}
    n(R_1+R_2)\leqslant & h(Y_1^n)-h(X_0^n+S_1^n+N_1^n, X_0^n+N_1^n)+h(Y_2^n)\nn\\
    =&h(Y_1^n)-h(S_1^n, X_0^n+N_1^n)+h(Y_2^n)\nn\\
    = &h(Y_1^n)-h(S_1^n)-h(X_0^n+N_1^n|S_1^n)+h(Y_2^n)\nn\\
\leqslant & h(Y_1^n)-h(S_1^n)-h(X_0^n+N_1^n|S_1^n,X_0^n)+h(Y_2^n)\nn\\
    \overset{(b)}{=}& h(X_0^n+X_1^n+S_1^n+N_1^n)-h(S_1^n)+h(X_0^n+X_2^n+N_1^n)-h(N_1^n)\nn\\
\leqslant& \frac{n}{2}\log2\pi e(P_1+P_0+\sum_{i=1}^nE(X_{0i}S_{1i})+Q+1)-\frac{n}{2}\log(2\pi eQ)\nn\\
&+\frac{n}{2}\log2\pi e(P_0+P_2+1)-\frac{n}{2}\log(2\pi e)\nn\\
  \leqslant& \frac{n}{2}\log\left(\frac{P_1+P_0+2\sqrt{P_0Q_1}+Q_1+1}{Q_1}\right)+\frac{n}{2}\log(P_0+P_2+1)\nn\\
\approx& \frac{n}{2}\log(P_0+P_2+1)\;\;(Q_1\rightarrow \infty)\nn
  \end{flalign}
where (b) follows from that $X_0^n$ and $S^n$ are independent from $N_1^n$.

\section{Proof of Proposition \ref{th:InnerSamec=0}}\label{apx:InnerSamec=0}

We use random codes and fix the following joint distribution:
\[P_{S_1U_1VX_0X_1X_2Y_1Y_2}= P_{VUS_1}P_{X_0|VUS_1}P_{X_1}P_{X_2}P_{Y_1|X_0 X_1 S_1}P_{Y_2|X_0X_2}.\]
Let $T_\epsilon^n(P_{S_1UVX_0X_1X_2Y_1Y_2})$ denote the strongly joint $\epsilon$-typical set based on the above distribution. For a given sequence $x^n$, let $T_\epsilon^n(P_{U|X}|x^n)$ denote the set of sequences $u^n$ such that $(u^n, x^n)$ is jointly typical based on the distribution $P_{XU}$.
\begin{enumerate}{\topsep=0.ex \leftmargin=0.25in
\rightmargin=0.in \itemsep =0.in}
  \item Codebook Generation
  \begin{itemize}
    \item Generate $2^{n(\tilde{R}_1)}$ codewords $u^n(t)$  with the probability of $P_U$, in which $t\in [1,2^{n\tilde{R}_1}]$.
    \item Generate $2^{n(\tilde{R}_2)}$ codewords $v^n(k)$  with the probability of $P_V$, in which $k\in [1,2^{n\tilde{R}_2}]$.
    \item Generate $2^{nR_1}$ codewords $x_1^{n}(w_1)$ with the probability of $P_{X_1}$, in which $w_1 \in [1,2^{nR_1}]$.
    \item Generate $2^{nR_2}$ codewords $x_2^{n}(w_2)$ with the probability of $P_{X_2}$, in which $w_2 \in [1,2^{nR_2}]$.
  \end{itemize}

  \item Encoding
  \begin{itemize}
    \item Encoder 0: Given $s_1^n$, find $\tilde{t}$, such that $(u^n(\tilde{t}),s_1^n)\in T_\epsilon^n(P_{S_1U})$. Such $u^n(\tilde{t})$ exists with high probability for large $n$ if
                      \begin{equation}\label{eq:InnerDMCproofc=0-1}
                            \tilde{R}_1 \geqslant I(S_1; U).
                      \end{equation}

    \item For each $\tilde{t}$ selected, select $\tilde{k}$, such that $(v^n(\tilde{k}), u^n(\tilde{t}),s_1^n) \in T_\epsilon^n(P_{VUS_1})$.
     Such $v^n(\tilde{k})$ exists with high probability for large $n$ if
                      \begin{equation}\label{eq:InnerDMCproofc=0-2}
                            \tilde{R}_2 \geqslant I(S_1U;V).
                      \end{equation}

    \item Map $(s_1^n, u^n, v^n)$ into $x_0^n$
    \item Encoder 1 and 2: Map $w_1$ into $x_1^n$, and map $w_2$ into $x_2^n$.
  \end{itemize}

  \item Decoding
  \begin{itemize}
    \item Decoder 1: Given $y_1^n$, find $(\hat{w}_1)$ such that $(x_1^{n}(\hat{w}_1), u^n(\hat{t}), y_1^n) \in T_\epsilon^n(P_{X_1UY_1})$. If no or more than one $\hat{w}_1$ can be found, declare an error. One can show that the decoding error is small for sufficient large $n$ if
               \begin{flalign}
                  R_{1} \leqslant I(X_1;Y_1U)\label{eq:InnerDMCproofc=0-3}\\
                  R_{1} + \tilde{R}_1 \leqslant I(X_1U;Y_1)\label{eq:InnerDMCproofc=0-4}
               \end{flalign}
\item Decoder 2:  Given $y_2^n$, find $(\hat{w}_2)$ such that $(x_2^{n}(\hat{w}_2), v^n(\hat{k}),  y_2^n) \in T_\epsilon^n(P_{X_2VY_2})$. If no or more than one $\hat{w}_2$ can be found, declare an error. One can show that the decoding error is small for sufficient large $n$ if
               \begin{flalign}
                  R_{2} &\leqslant I(X_2;Y_2V)\label{eq:InnerDMCproofc=0-5}\\
                  R_{2} + \tilde{R}_2 &\leqslant I(X_2 V;Y_2)\label{eq:InnerDMCproofc=0-6}
               \end{flalign}
  \end{itemize}
\end{enumerate}

According to \eqref{eq:InnerDMCproofc=0-1}-\eqref{eq:InnerDMCproofc=0-6}, exploit the Foriour-Mozkin elimination to eliminate $\tilde{R}_1$ and $\tilde{R}_2$, and we have the desired achievable region.

\section{Proof of Proposition \ref{th:OuterIndepend2}}\label{apx:OuterIndepend2}
First of all, $R_1$ and $R_2$ is bounded by single rate bound respectively.

For the sum rate bound, we start from the Fano's inequality
\begin{flalign}
    n(R_1+R_2)\leqslant &I(W_1;Y_1^n)+I(W_2;Y_2^n)\nn\\
=&h(Y_1^n)-h(Y_1^n|W_1)+h(Y_2^n)-h(Y_2^n|W_2)\nn\\
\overset{(a)}{=}&h(Y_1^n)-h(Y_1^n|W_1 X_1^n)+h(Y_2^n)-h(Y_2^n|W_2 X_2^n)\nn\\
=&h(Y_1^n)-h(X_0^n+S_1^n+N_1^n)+h(Y_2^n)-h(X_0^n+S_2^n+N_2^n)\nn\\
\leqslant & h(Y_1^n)-h(X_0^n+S_1^n+N_1^n| X_0^n+N_1^n)\nn\\
    &+h(Y_2^n)-h(X_0^n+S_2^n+N_2^n| X_0^n+N_2^n, X_0^n+S_1^n+N_1^n)\nn\\
    &+h(X_0^n+N_1^n)-h(X_0^n+N_1^n)\nn
  \end{flalign}
where (a) follows from that $X_1^n$ is function of $W_1$, and $X_2^n$ is function of $W_2$, and they are independent from $X_0^n$, state and noise. Because the two decoders decode based on the marginal distribution only, setting $N_1^n=N_2^n$ does not influence the channel capacity, therefore,
\begin{flalign}
    n(R_1+R_2)\leqslant & h(Y_1^n)-h(X_0^n+S_1^n+N_1^n, X_0^n+S_2^n+N_1^n, X_0^n+N_1^n)\nn\\
    &+h(Y_2^n)+h(X_0^n+N_1^n)\nn\\
    =&h(Y_1^n)-h(S_1^n, S_2^n, X_0^n+N_1^n)+h(Y_2^n)+h(X_0^n+N_1^n)\nn\\
    = &h(Y_1^n)-h(S_1^n)-h(S_2^n)-h(X_0^n+N_1^n|S_1^n,S_2^n)\nn\\
    &+h(Y_2^n)+h(X_0^n+N_1^n)\nn\\
\leqslant & h(Y_1^n)-h(S_1^n)-h(S_2^n)-h(X_0^n+N_1^n|S_1^n,S_2^n,X_0^n)\nn\\
    &+h(Y_2^n)+h(X_0^n+N_1^n)\nn\\
    \overset{(b)}{=}& h(X_0^n+X_1^n+S_1^n+N_1^n)-h(S_1^n)-h(S_2^n)-h(N_1^n)\nn\\
    &+h(X_0^n+X_2^n+S_2^n+N_1^n)+h(X_0^n+N_1^n)\nn\\
\leqslant& \frac{n}{2}\log2\pi e(P_1+P_0+\sum_{i=1}^nE(X_{0i}S_{1i})+Q_1+1)-\frac{n}{2}\log(2\pi eQ_1)-\frac{n}{2}\log(2\pi eQ_2)\nn\\
&-\frac{n}{2}\log(2\pi e)+\frac{n}{2}\log2\pi e(P_2+P_0+\sum_{i=1}^nE(X_{0i}S_{2i})+Q_2+1)+\frac{n}{2}\log2\pi e(P_0+1)\nn\\
  \leqslant& \frac{n}{2}\log\left(\frac{P_1+P_0+2\sqrt{P_0Q_1}+Q_1+1}{Q_1}\right)+\frac{n}{2}\log\left(\frac{P_2+P_0+2\sqrt{P_0Q_2}+Q_2+1}{Q_2}\right)\nn\\
    &+\frac{n}{2}\log(P_0+1)\nn\\
\approx& \frac{n}{2}\log(P_0+1)\;\;(Q_1\rightarrow \infty, Q_2\rightarrow \infty)\nn
  \end{flalign}
where (b) follows from that $X_0^n$, $S_1^n$ and $S_2^n$ are independent from $N_1^n$.

  \section{Proof of Proposition \ref{th:IndependC}}\label{apx:IndependentC}
   The theorem can be proved in two parts, 1. if $P_1+P_2\geqslant P_0+1$, the sum capacity is obtained; 2. characterize $\gamma$ such that for this time allocation, the rate achieved is on the capacity boundary.

  1. For a certain $P_0$, we consider the following two cases.

   a). If the power constraint satisfies $P_1+P_2=P_0+1$, by following Proposition \ref{th:InnerIndepend2}, setting $\gamma=\frac{P_1}{P_1+P_2}$ , the point $(R_1,R_2)=(\frac{P_1}{2(P_1+P_2)}\log(1+P_0),\frac{P_2}{2(P_1+P_2)}\log(1+P_0))$ is achieved, which is also on the outer bound in Proposition \ref{th:OuterIndepend2}.

   b). If $P_1+P_2\geqslant P_0+1$, the outer bound does not change for the same $P_0$, we use the power $\tilde{P}_1+\tilde{P}_2=P_0+1$ and $\tilde{P}_1\leqslant P_1$, $\tilde{P}_2\leqslant P_2$ for each transmitter and obtain the sum capacity as concluded in a).

  2. We start from considering the constraint for $\gamma$ under which the rate pair $(R_1,R_2)$ achieves the sum capacity, i.e.
  \begin{flalign}
    \frac{P_2}{1-\gamma}&\geqslant P_0+1\label{eq:gammacondition-1}\\
    \frac{P_1}{\gamma}&\geqslant P_0+1\label{eq:gammacondition-2}.
  \end{flalign}

It is clear that \eqref{eq:gammacondition-1} implies
\begin{flalign}
    \gamma&\geqslant 1-\frac{P_2}{P_0+1}, \nn
  \end{flalign}
and \eqref{eq:gammacondition-2} implies
\begin{flalign}
    \gamma&\leqslant \frac{P_1}{P_0+1}.\nn
  \end{flalign}
Considering $0 \leq \gamma \leq 1$, we obtain the desired bounds on $\gamma$.

\section{Proof of Proposition \ref{th:OuterIndependKW0}}\label{apx:OuterIndependKW0}
 The individual rate is first bounded by the point-to-point channel capacity, respectively.

We then bound the sum rate. By following the Fano's inequality, we have
\begin{flalign}
    \sum_{k=0}^K nR_k\leqslant &\sum_{k=0}^K I(W_k;Y_k^n)\nn\\
   =&\sum_{k=0}^K [ h(Y_k^n)-h(Y_k^n|W_k)]\nn\\
\overset{(a)}{=}& h(Y_0^n)-h(X_0^n+N_0^n|W_0)+\sum_{k=1}^K [h(Y_k^n)-h(Y_k^n|W_k X_k^n)]\nn\\
=&h(Y_0^n)-h(X_0^n+N_0^n|W_0)+\sum_{k=1}^K[h(Y_k^n)-h(X_0^n+S_k^n+N_k^n)]\nn\\
\leqslant & h(Y_0^n)-h(X_0^n+N_0^n|W_0)\nn\\
&+\sum_{k=1}^K [h(Y_k^n)-h(X_0^n+S_k^n+N_k^n| X_0^n+N_0^n,W_0, X_0^n+S_{k-1}^n+N_{k-1}^n,\dots,X_0^n+S_1^n+N_1^n)]\nn\\
=& [\sum_{k=0}^K h(Y_k^n)]-h(X_0^n+S_1^n+N_1^n,\dots, X_0^n+S_K^n+N_K^n, X_0^n+N_0^n|W_0)\nn
  \end{flalign}
  where (a) follows from that $X_k^n$ is function of $W_k$, and they are independent from $X_0^n$, state and noise. Because the decoders decode based on the marginal distribution only, we can set $N_k^n=N_0^n$ for $k=1,\dots,K$, therefore,
\begin{flalign}
    \sum_{k=1}^K n R_k\leqslant &[\sum_{k=0}^K h(Y_k^n)]-h(S_1^n,\dots, S_K^n, X_0^n+N_0^n)\nn\\
    = &h(Y_0^n)+\sum_{k=1}^K[h(Y_k^n)-h(S_k^n)]-h(X_0^n+N_0^n|S_1^n,\dots,S_K^n,W_0)\nn\\
\leqslant &h(X_0^n+N_0^n)+\sum_{k=1}^K[h(Y_k^n)-h(S_k^n)]-h(X_0^n+N_0^n|S_1^n,\dots,S_K^n,W_0, X_0^n)\nn\\
    \overset{(b)}{=}& h(X_0^n+N_0^n)+\sum_{k=1}^K[h(X_0^n+X_k^n+S_k^n+N_k^n)-h(S_k^n)]-h(N_0^n)\nn\\
\leqslant & \sum_{k=1}^K[\frac{n}{2}\log2\pi e(P_k+P_0+\sum_{i=1}^nE(X_{0i}S_{ki})+Q_k+1)-\frac{n}{2}\log(2\pi eQ_k)]\nn\\
&-\frac{n}{2}\log(2\pi e)+\frac{n}{2}\log2\pi e(P_0+1)\nn\\
  \leqslant& \sum_{k=1}^K\frac{n}{2}\log\left(\frac{P_k+P_0+2\sqrt{P_0Q_k}+Q_k+1}{Q_k}\right)+\frac{n}{2}\log(P_0+1)\nn\\
\rightarrow& \frac{n}{2}\log(P_0+1)\;\;\text{as }Q_k\rightarrow \infty\nn
\end{flalign}

\bibliographystyle{IEEEtran}
\bibliography{ZChannel_Journal}

\end{document}